\documentclass[12pt]{amsart}
\usepackage[letterpaper, margin=1in]{geometry}
\usepackage[utf8]{inputenc} 
\usepackage[T1]{fontenc}    
\usepackage{url}            
\usepackage{booktabs}       
\usepackage{amsfonts}       
\usepackage{nicefrac}       
\usepackage{xcolor}         
\usepackage[flushleft]{threeparttable}
\usepackage{siunitx}
\usepackage{caption}

\usepackage{setspace}
\usepackage{natbib}
\usepackage{amsmath}
\usepackage{amsthm}
\usepackage{amssymb}
\usepackage{todonotes}
\usepackage{multirow}
\usepackage{comment}
\usepackage{enumerate}
\usepackage[shortlabels]{enumitem}

\usepackage[linesnumbered,ruled]{algorithm2e}
\SetKwInput{KwInput}{Input}                
\SetKwInput{KwOutput}{Output}              

\usepackage{hyperref}
\hypersetup{
    colorlinks=true,
    linkcolor=red,
    citecolor=black,
    filecolor=magenta,      
    urlcolor=cyan,
}

\theoremstyle{plain}
\newtheorem{thm}{\protect\theoremname}
\newtheorem{lem}{Lemma}

\newtheorem*{asm*}{Assumption}
\providecommand{\theoremname}{Theorem}

\newtheorem{cor}{Corollary}

\theoremstyle{remark}
\newtheorem{remark}{Remark}

\newcommand{\bG}{\mbox{\boldmath $G$}}

\newcommand{\bbeta}{\mbox{\boldmath $\beta$}}

\newcommand{\bOmega}{\mbox{\boldmath $\Omega$}}
\newcommand{\bW}{\mbox{\boldmath $W$}}

\newcommand{\var}{\operatorname{var}}

\newcommand{\opn}[1]{\left\|{#1}\right\|}
\newcommand{\Avar}{\operatorname{Avar}}

\newcommand{\sub}{\operatorname{sub}}

\newcommand{\MS}[1]{\todo[color=red!20!white,fancyline,size=\tiny]{#1}{}}

\title[SGMM]{SGMM: Stochastic Approximation to Generalized Method of Moments}

\author[Chen]{Xiaohong Chen}
\address{Department of Economics, Yale University, New Haven, CT 06520, USA  and Cowles Foundation} 
\email{xiaohong.chen@yale.edu}  

\author[Lee]{Sokbae  Lee}
\address{Department of Economics, Columbia University, New York, NY 10027, USA}
\email{sl3841@columbia.edu} 

\author[Liao]{Yuan  Liao}
\address{Department of Economics, Rutgers University, New Brunswick, NJ 08901, USA}
\email{yuan.liao@rutgers.edu}

\author[Seo]{Myung Hwan  Seo}
\address{Department of Economics, Seoul National University, Seoul, 08826, Korea}
\email{myunghseo@snu.ac.kr}

\author[Shin]{Youngki  Shin}
\address{Department of Economics, McMaster University, Hamilton, ON L8S 4L8, Canada}
\email{shiny11@mcmaster.ca}

\author[Song]{Myunghyun Song}
\address{Department of Economics, Columbia University, New York, NY 10027, USA}
\email{ms6347@columbia.edu}

\date{October 2023}

\thanks{
The present paper is based on the JFEC Invited Lecture given by Xiaohong Chen at the SOFIE 15th Annual Conference at Sungkyunkwan University, June 17, 2023 in Seoul, South Korea. We thank discussants at the conference (Nour Meddahi and Joon Park) for their helpful comments. This article is completed during Chen's and Lee’s visit to the Korean Bureau of Economic Research (KBER) in Seoul National University, whose hospitality and the support by the BK21 FOUR (Fostering Outstanding Universities for Research) funded by the Ministry of Education (MOE, Korea) and National Research Foundation of Korea (NRF) is greatly acknowledged. Financial support from KBER and Innovation in the Institute of Economic Research, Seoul National University is greatly acknowledged. This research was enabled in part by support provided by Compute Ontario (\url{https://com puteontario.ca}) and the Digital Research Alliance of Canada (\url{https://alliancecan.ca}). Chen greatly acknowledges partial support from the Cowles Foundation.}

\begin{document}








\maketitle

\begin{abstract}
We introduce a new class of algorithms, Stochastic Generalized Method of Moments (SGMM), for estimation and inference on (overidentified) moment restriction models. Our SGMM is a novel stochastic approximation alternative to the popular \cite{GMM} (offline) GMM, and offers fast and scalable implementation with the ability to handle streaming datasets in real time. We establish the almost sure convergence, and the (functional) central limit theorem for the inefficient online 2SLS and the efficient SGMM. Moreover, we propose online versions of the Durbin-Wu-Hausman and Sargan-Hansen tests that can be seamlessly integrated within the SGMM framework. Extensive Monte Carlo simulations show that as the sample size increases, the SGMM matches the standard (offline) GMM in terms of estimation accuracy and gains over computational efficiency, indicating its practical value for both large-scale and online datasets. We demonstrate the efficacy of our approach by a proof of concept using two well known empirical examples with large sample sizes.
\end{abstract}


\onehalfspacing

\section{Introduction}

Machine learning techniques have revolutionized the analysis of vast and unconventional datasets. Among them, stochastic approximation (SA) or more commonly called stochastic gradient descent (SGD) pioneered by \citet{Robbins-Monro-1951} has proven highly valuable due to its computational simplicity and scalable online implementation.
In econometrics, Halbert White was a great trailblazer of SGD. For example, \cite{white1989jasa} applied earlier general theory on the almost sure consistency and asymptotic normality of recursive nonlinear least squares (NLS) to parametric single-hidden layer artificial neural network (ANN) regression models with independent and identically distributed (iid) data; \cite{kuanwhite1994er} developed asymptotic theory for general nonlinear models of weakly dependent processes, including applications to nonlinear regression via neural networks;\footnote{\cite{Pastorello:Patilea:Renault} applied the results of \cite{kuanwhite1994er} to obtain the consistency and asymptotic normality for their recursive latent backfitting procedure in a just-identified moment problem.} 
\cite{Chen:White:1998:JET} applied stochastic approximation to bounded rationality learning;
\cite{chen2002asymptotic} established asymptotic theory of SGD for Hilbert space-valued mixingale, dependent error processes.

While traditionally used for computational purposes, such as optimizing objective functions
\citep[see, e.g.][for its review]{bottou2018optimization}, 
SGD has also received attention for its statistical properties. As an early path-breaking work, \citet{polyak1992acceleration} obtained conditions under which an average of the SGD sequence is asymptotically normal with mean zero and an efficient variance matrix in parametric regressions. A more recent literature on SGD covers diverse topics: 
regularized methods for high-dimensional M-estimators \citep{agarwal2010fast};
implicit SGD \citep{Toulis2017,Won:lab:2022}; 
moment-adjusted SGD \citep{liang2017statistical};
non-asymptotic results for the averaged SGD \citep{pmlr-v99-anastasiou19a,pmlr-v125-mou20a};
among many others.
A branch of the recent literature is concerned with online statistical inference:
bootstrap \citep{fang2018online};
batch-means \citep{chen2020statistical,zhu2021online};
random scaling \citep{ChenLai2021,lee2021fast,li2021statistical} among other possible modes of inference \citep[e.g.,][]{Chee-Kim-Toulis:23}.
The studies by date, however,  have mainly focused on M-estimation. That is, the SGD has been mainly used  for estimating a parameter of interest $\beta^*$ that is identified as the unique minimizer of a population loss function $\min_{\beta} \mathbb{E}\left[ \ell_i(\beta) \right]$, where
$\ell_i(\beta)$ is a known real-valued function of the $i$-th observation and a parameter $\beta$. We therefore refer to the usual SGD-type estimators as ``M-type SGD'', which takes the following basic form:
\[
\beta_{i}=\beta_{i-1}-\gamma_{i}\frac{\partial}{\partial \beta} \ell_i \left(\beta_{i-1}\right),~~\text{for some~}\gamma_{i}\searrow 0~.
\]

In applications in economics and finance, we often encounter a different type of estimation problems, the so-called ``Z-estimation,'' where the parameter of interest is identified as a unique solution to a set of moment conditions, i.e., 
$\mathbb{E}\left[ g_i(\beta_*) \right]=0$, where
$g_i(\beta)$ is a known function of the $i$-th observation and a parameter $\beta$. Here,  
$\beta_*$ denotes the unique solution.
These moment conditions,  under just (or exact) identification  $\dim(\mathbb{E}\left[ g_i(\beta) \right]) = \dim (\beta)$, would yield the ``estimating equations''  for the parameter of interest $\beta_*$. Most importantly, the popular (offline) generalized methods of moments (GMM) of \cite{GMM} allows for overidentified moment restrictions in the sense that $\dim(g_i) > \dim (\beta)$, and that efficient estimation of $\beta_*$ and model-specification test can be carried out using the same optimally weighted GMM loss function $\min_{\beta} \left(\mathbb{E}[ g_i(\beta)]'\left[\var(g_i(\beta^*))\right]^{-1}\mathbb{E}[ g_i(\beta)]\right)$.
Unlike M-type SGD, it is unclear how to obtain a SGD alternative to the optimally weighted GMM and to establish its statistic properties. This is especially the case for the overidentified moment restriction models.

In this paper, we develop new stochastic approximation methods for GMM, allowing for possibly overidentified moment restriction models. As a premier example, we focus on linear instrumental variable  (IV) regression, where the moment restrictions are linear in the parameters of interest. Despite of being restricted to linearity, this type of models are widely applicable in economics and finance applications. We argue that aside from the more traditional IV estimators (e.g., two-stage least squares), the SA-based estimation is a natural option for IV regression because of the following reasons. First, it  is  fully capable of handling problems of very large  datasets. Because by nature of SA, the estimation is updated one-observation-at-a-time, so it is suitable for online learning.  Second, it is convenient to work  with the moment conditions of econometric models, the essence of possibly overidentified Z-estimation. Hence, we view our approach as a highly scalable estimation and inference method for the moment restriction models. 

We first propose a stochastic approximation to the two-stage least squares (2SLS) and analyze its stochastic properties.
We provide conditions under which our SA-based estimator is first-order asymptotically equivalent to the standard (offline) 2SLS estimator.
The inference problems studied in this paper based on the SA-based 2SLS estimator include: obtaining confidence interval for $\beta_*$  as well as testing for the validity of the specified instruments. For the former problem, we employ the recently developed random scaling inference
 of  \cite{lee2021fast}, which is fast, suitable for online learning, and easily adaptable for subvector inference. For the latter problem, we develop an ``online'' version of the  Durbin-Wu-Hausman test by comparing the probability limits of the OLS and 2SLS estimators, both obtained using the SA-based methods. In both problems, because the pivotal statistics are scaled by a random matrix similar to the ``fixed-b'' smoothing, the asymptotic distributions are mixed normal, whose critical values have been  tabulated in the literature, and are readily available for statistical inference. See, e.g., \citet{kiefer2000simple,velasco2001edgeworth,sun2008optimal,sun2013ej,Chen:Liao:Sun:2014,lazarus2018har,gupta2021robust} for related papers in the time series literature.

 As in the regular GMM-estimation, one of the central problems is the efficient estimation of $\beta_*$ using optimally weighted moment conditions. We show that the optimal weighting is also naturally incorporated by the overidentified Z-type SA algorithm, where we sequentially update the optimal weighting matrix along the path of the SA iteration. Despite of sequentially updating an inverse of a covariance matrix,  we show that implementation is still fast because it is based on the Sherman–Morrison–Woodbury (SMW) formula; in other words, our implementation does not involve an actual high-dimensional matrix inversion. In theory, we show that the optimally weighted SA-based estimator, termed stochastic GMM (SGMM), is first-order asymptotically equivalent to the well known (offline) two-step efficient GMM estimator of the IV regression model. 
 As by-products, we provide online plug-in optimal inference on $\beta_*$ as well as an online version of the Sargan-Hansen specification test using the efficient SGMM estimator.

The literature on stochastic approximation to 2SLS or GMM is almost non-existent. Some exceptions are 
\citet{venkatraman2016online} and \citet{della2023online}.
\citet{venkatraman2016online} proposed to use fitted values from the first stage to build an online algorithm;
\citet{della2023online} considered the just-identified IV estimator to study online regression as well as the bandit problem. 
The aforementioned papers carried out some sort of regret analyses but none of them focused on statistical properties of their proposed methods.  

The remainder of the paper is organized as follows.
Section~\ref{sec:baseline-weight} outlines our basic algorithm and its theoretical properties.
Section~\ref{sec:eff:est} provides an efficient online algorithm that is first-order asymptotically equivalent to efficient GMM estimators.
In Section~\ref{sec:inference}, we show that 
online versions of the Durbin-Wu-Hausman and Sargan-Hansen tests can be seamlessly integrated within the SGMM framework.
Section~\ref{sec:Monte-Carlo} reports the results of extensive Monte Carlo experiments
and Section~\ref{sec:examples} provides empirical examples based on two well known studies: \citet{angrist1991does} and \citet{angrist1998children}. 
Section~\ref{sec:sgmm:extension} discusses extensions, including an extension to Nonlinear SGMM. 
Section~\ref{sec:appendix} provides all the proofs of the theoretical results in the main text.

\textbf{Notation}. We denote the $\ell^2$-vector norm of $x$ in $\mathbb R^n$ by $\|x\|_n = (\sum_{i=1}^n |x_i|^2)^{1/2}$, and the $\ell^2$-operator norm of an $n$ by $m$ matrix by $\|{A}\|_{\operatorname{op}} := \sup_{\|x\|_m\le 1}\|A x\|_n$. We will occasionally write the $\ell^2$-vector and operator norm as $\|x\|$ or $\|A\|$ suppressing the subscripts when there is no possibility of confusion. The $\ell^2$-vector norm is equal to the operator norm when vectors are seen as $n\times 1$ matrices. 
%
%
Let $(S, d)$ be a complete metric space. We denote the weak convergence of $S$-valued random variables $X_n$ by $X_n \rightsquigarrow X$, where $X$ denotes the weak limit.
In addition, $\overset{d}{\to}$ refers to convergence in distribution.
For a real sequence $a_n$ and positive numbers $b_n$, we write $a_n = O(b_n)$ or $a_n \lesssim b_n$ if there exists a uniform constant $C>0$ such that $|a_n|\le C b_n$ holds for all $n$. For a sequence of real r.v.'s, we also denote $X_n = o_{a.s.}(b_n)$ or $X_n = O_{a.s.}(b_n)$ if $\lim_{n\to \infty}|X_n|/b_n = 0$ or $\limsup_{n\to\infty} |X_n|/b_n < \infty$ with probability 1, respectively. We use $X_n = O_P(1)$ and $X_n = o_P(1)$ to indicate that $X_n$ is a tight sequence of random variables and converges to $0$ in probability, respectively.

\section{Stochastic Approximation for Instrumental Variable Regression}\label{sec:baseline-weight}

\subsection{Model}

Consider a linear instrumental variables regression model
\begin{align}\label{model}
y_i = x_i'\beta_* + u_i,~~~ \mathbb{E}\left[u_i z_i \right]=0,
\end{align}
where $y_i \in \mathbb{R}$ is the dependent variable,
$x_i \in \mathbb{R}^{d_\beta}$ a vector of covariates and some of which are endogenous in the sense that  $\mathbb{E}\left[x_i u_i\right]\neq 0$, $u_i \in \mathbb{R}$ is the regression error, $z_i \in \mathbb{R}^{d_g}$
is a vector of instrumental variables, and $\beta_* \in \mathbb{R}^{d_\beta}$ is a vector of unknown true parameters of interest. Let $g_i(\beta) \equiv z_i \left( x_i'\beta - y_i \right)$ and assume $d_g \geq d_{\beta}$. We focus on estimation of $\beta_*$ using the following linear moment restriction models:
\begin{align*}
\mathbb{E}\left[ g_i(\beta_*) \right]=0.  
\end{align*}
Throughout the paper, we let $G_i \equiv z_i x_i'$ and $H_i \equiv -z_i y_i$; hence,
$g_i(\beta) = G_i \beta + H_i$. We also let $G \equiv \mathbb E [G_i]$ and $H \equiv \mathbb E [H_i]$. That is,
a letter without subscript denotes its expectation. The linear instrumental variable regression model \eqref{model} becomes 
$G \beta_* + H = 0$.

\subsection{S2SLS Algorithm}

In this section, we propose a new stochastic approximation algorithm to estimate $\beta_*$.  Let $\mathbb S \equiv \{\mathcal{D}_i  = (x_i, z_i, y_i)\}_{i=1}^{n}$ 
be an i.i.d. sample of size $n$ drawn from a population distribution satisfying model \eqref{model}. Let $\mathbb S_0 \equiv \{\mathcal{D}_{0j}= (x_{0j}, z_{0j}, y_{0j})\}_{j=1}^{n_0}$ be an initialization random sample of size $n_0$ drawn from model \eqref{model}, with $n_0 \ll n$. Denote $\mathcal F_i = \sigma(\mathbb S_0 \cup \{\mathcal{D}_j\}_{j=1}^{i})$ for $i \ge 0$. Compute the initial estimator $\beta_0 \in \sigma(\mathbb S_0)=\mathcal F_0$ using 2SLS, GMM or any other estimation methods.\footnote{In fact, our asymptotic theory allows for any arbitrary choice of the initial estimator, including $\beta_0 = 0$. The finite sample performance depends on the quality of the initial estimator, however.} 
Let $\Phi_0 \equiv \frac 1 {n_0} \sum_{j=1}^ {n_0} z_{0j} x_{0j}' \in \mathcal F_0$ and $W_0 \equiv \left( \frac{1}{n_0} \sum_{j=1}^{n_0} z_{0j} z_{0j}' +  \eta_0 I \right)^{-1}$ for a fixed constant $\eta_0 \geq 0$.\footnote{In many applications, the choice of $\eta_0 = 0$ will suffice, provided that $n_0$ is large enough and $z_i$ is linearly independent.} Starting from $\left(\beta_0, \Phi_0, W_0\right)$, we update 
 the stochastic process $(\beta_i, \Phi_i, W_i)_{i=1}^\infty$ sequentially as 
\begin{subequations}\label{algo:siv:2sls}
\begin{align}
    \beta_i &= \beta_{i-1} - \gamma_i (\Phi_{i-1}' W_{i-1} \Phi_{i-1})^{\dagger} \Phi_{i-1}' W_{i-1} g_i(\beta_{i-1}), \label{algo:beta:2sls} \\ 
    \Phi_i &= \frac{n_0 + i - 1}{ n_0 + i} \Phi_{i-1} + \frac{1} { n_0 + i} G_i, \label{algo:Phi:2sls} \\
            m_i&=n_0+i-1+z_i'W_{i-1}z_i, \label{algo:m:W:2sls} \\  
        W_i&= \frac{n_0+i}{n_0+i-1} W_{i-1} \left[I- m_i^{-1} z_i z_i'W_{i-1}\right], \label{algo:beta:W:2sls} \\    
    \bar \beta_i &= \frac{i-1 }{ i} \bar \beta_{i-1} + \frac{1} { i}\beta_i,
    \label{algo:betabar:2sls}
\end{align}
\end{subequations}
where $g_i(\beta_{i-1}) = G_i\beta_{i-1} + H_i$,
$\bar \beta_0 = \beta_0$ and 
$\gamma_i \equiv \gamma_0 i^ {-a}$ 
is a learning rate with some predetermined constants  $\gamma_0 > 0$ and $a \in (1/2, 1)$.  
Here, $A^{\dagger}$ denotes the generalized inverse of a matrix $A$.
We propose to use $\bar\beta_n$, which is called the \citet{Polyak1990}-\citet{ruppert1988efficient} average, as an estimator of $\beta_*$. 

\begin{remark}[Sherman-Morrison-Woodbury matrix inversion]\label{rem:woodbury:W}
Note that we update the $d_g \times d_g$- weighting matrix $W_{i}$ sequentially in the S2SLS algorithm \eqref{algo:siv:2sls}.
To convey the idea behind this updating rule, let $Q \equiv \mathbb{E}[ z_i z_i']=W^{-1}$, and
consider an updating rule for $Q_i$:
\begin{align*}
Q_i &= \frac{n_0 + i - 1}{ n_0 + i} Q_{i-1} + \frac{1} { n_0 + i} z_i z_i'
\end{align*}
and update $\beta_i$ using the inverse of $Q_i$:
\begin{align*}
\beta_i &= \beta_{i-1} - \gamma_i (\Phi_{i-1}'Q_{i-1}^{-1}\Phi_{i-1})^{\dagger}\Phi_{i-1}'Q_{i-1}^{-1} g_i(\beta_{i-1}).    
\end{align*}
The proposed algorithm in \eqref{algo:siv:2sls} explicitly computes $Q_{i-1}^{-1}=W_{i-1}$ 
using Sherman-Morrison-Woodbury (SMW) formula, showing that it is unnecessary to compute the inverse of $Q_{i-1}$ each time but it suffices to update the scalar quantity $m_i$ and the matrix $W_i$ accordingly. 
The positive constant $\eta_0$ in the definition of $W_0$ ensures that $W_{i-1}$ is well defined for all $i$.
\end{remark}

\begin{remark}[Computation of $(\Phi_{i-1}' W_{i-1} \Phi_{i-1})^{\dagger}$]\label{rem:woodbury:1}
When the dimension $d_\beta$ of $\beta$ is high, it would be time-consuming to directly compute
$(\Phi_{i-1}' W_{i-1} \Phi_{i-1})^{\dagger}$ in \eqref{algo:beta:2sls}.
Under the identification assumption given below, 
$(\Phi_{i-1}' W_{i-1} \Phi_{i-1})$ is invertible with probability approaching one,
and hence its inverse can be computed via SMW formula.
Specifically, let $\mathcal{H}_i := (\Phi_{i}'W_{i} \Phi_{i})^{-1}  \in \mathbb{R}^{d_{\beta} \times d_{\beta}}$, we have:
\begin{align}\label{Woodbury:H:matrix}
    \mathcal{H}_i = \frac {n_0 + i} {n_0 + i-1} \left( \mathcal{H}_{i-1} - \mathcal{H}_{i-1} \mathcal{U}_i (\mathcal{D}_i + \mathcal{U}_i'\mathcal{H}_{i-1}\mathcal{U}_i )^{-1} \mathcal{U}_i' \mathcal{H}_{i-1}\right),  
\end{align}
where 
\begin{align*}
    \mathcal{D}_i &=\operatorname{diag}\left[-m_i, n_0 + i-1 \right]   \in \mathbb{R}^{2 \times 2},\\
    \mathcal{U}_i &= [x_i-\Phi_{i-1}' W_{i-1}z_i, x_i] \in \mathbb{R}^{d_{\beta} \times 2}.
\end{align*}
Because $\mathcal{D}_i + \mathcal{U}_i' \mathcal{H}_{i-1} \mathcal{U}_i$ is a 2 by 2 matrix, using \eqref{Woodbury:H:matrix} would be computationally advantageous when $d_\beta \gg 2$. 
\end{remark}

\subsection{Intuition Behind Our Algorithm}

One difficulty in constructing an stochastic approximation algorithm for an IV regression is that the model \eqref{model} is possibly an overidentified moment restrictions; therefore, there is no obvious form of stochastic gradient descent for an IV regression.

Suppose that $G = \mathbb{E} [ z_i x_i']$ has full rank $d_\beta$, which is the standard assumption for IV regression. Then $ \lambda_{\min}(G' G) > 0$, where $\lambda_{\min}(\cdot)$ is the smallest eigenvalue.
We propose to build our algorithm based on a stochastic approximation of the ordinary differential equation: 
\begin{align*}
\dot \beta = -(\beta - \beta_*) = - (G' W G)^{-1} G' W G(\beta - \beta_*),    
\end{align*}
where $W$ is a positive definite symmetric weighting matrix. 
Note that 
\begin{align*}
    \mathbb E[\beta_i -\beta_{i-1}\mid \mathcal F_{i-1}]
    &= -\gamma_i\mathbb E[(\Phi_{i-1}' W_{i-1} \Phi_{i-1})^{\dagger} \Phi_{i-1}' W_{i-1} g_i(\beta_{i-1})\mid \mathcal F_{i-1}] \\
    &= -\gamma_i (\Phi_{i-1}' W_{i-1} \Phi_{i-1})^{\dagger} \Phi_{i-1}' W_{i-1} (G\beta_{i-1} + H) \\
    &= -\gamma_i (\Phi_{i-1}' W_{i-1} \Phi_{i-1})^{\dagger} \Phi_{i-1}' W_{i-1} G(\beta_{i-1}-\beta_*).
\end{align*}
Thus, our updating rule on average moves in the direction which reduces the difference between the current state $\beta_{i-1}$ and the true value $\beta_*$, provided that 
$\Phi_{i-1}$ and $W_{i-1}$ are close to $G$ and $W$ to make $\Phi_{i-1}'W_{i-1}G$ positive definite. The latter requirement is satisfied for sufficiently large $i$ due to the law of large numbers.
When $n$ is large enough,
$(\Phi_{i-1}' W_{i-1} \Phi_{i-1})^{\dagger} \Phi_{i-1}' W_{i-1} G$ is close to an identity matrix.

\begin{remark}
It is important to pre-multiply the scaling matrix
$(\Phi_{i-1}' W_{i-1} \Phi_{i-1})^{\dagger}$ in \eqref{algo:beta:2sls}, so that 
the scale of observations can be automatically adjusted.
From this perspective, 
    our algorithm can be regarded as a sort of \emph{second-order method} using the terminology of
    \citet[Section 6]{bottou2018optimization}. They emphasize that SGD or the batch gradient method, which can be called
     first-order methods, are not scale invariant. 
     Another perspective, which is more intimately tied to asymptotic theory, is that our algorithm is based on an influence function for 2SLS and GMM. In other words, our formulation of the second-order matrix is reverse-engineered to reproduce the same asymptotic variances of offline 2SLS and GMM.
\end{remark}

\begin{remark}
We multiplied $(\Phi_{i-1}' W_{i-1} \Phi_{i-1})^{\dagger} \Phi_{i-1}' W_{i-1}$ before $g_i(\beta_{i-1})=G_i\beta_{i-1}+H_i$ in \eqref{algo:beta:2sls}, as opposed to $(\Phi_{i}' W_{i} \Phi_{i})^{\dagger} \Phi_i' W_{i}$. If the latter were multiplied instead of the former, it would have introduced $O(\gamma_i /i)$-bias in each update, whose exact magnitude depends on the data-generating process. 
Furthermore, asymptotic analysis would have been more complicated. 
\end{remark}

\begin{remark}
It is noteworthy that the updating rule depends only on the relative location to the true value $\Delta_{i-1} = \beta_{i-1} - \beta_*$, but not directly through the location of $\beta_{i-1}$. This guarantees that we can develop asymptotic theory by assuming that $\beta_* = 0$ without loss of generality. 
\end{remark}

\begin{remark}[IV Clustered Dependence]
We can extend IV model \eqref{model} to a cluster-dependent setting:
\begin{align}\label{model:cluster}
y_{i,t} = x_{i,t}'\beta_* + u_{i,t},~~~ \mathbb{E}\left[u_{i,t} z_{i,t} \right]=0, \quad  t=1, \ldots, T_i,  i=1, \ldots, n,
\end{align}
where there are $n\to\infty$ clusters, and  within each cluster we have finite many ($T_i$) observations.
It is straightforward to accommodate clustered dependence. 
Specifically,  we now update the estimator at the  $i$ level: 
run a modified version of \eqref{algo:siv:2sls} with 
\begin{align*}
G_i = \frac{1}{T_i} \sum_{t=1}^{T_i} z_{i,t} x_{i,t}' \; \text{ and } \; 
H_i = - \frac{1}{T_i} \sum_{t=1}^{T_i} z_{i,t} y_{i,t}.
\end{align*}
In words, instead of updating the estimate for each observation, we treat all individuals within a cluster as a ``mini-batch'' and update the estimate in batches. 
This allows for arbitrary dependence within clusters.
\end{remark}

\begin{remark}[Two Sample IV]
It is easy to accommodate our algorithm for the case when 
$G_i$ and $H_i$ are from two different datasets.    
\end{remark}


\subsection{Asymptotic Properties}

We first state basic regularity conditions.

\begin{asm*}\label{asm:conditions}
\begin{enumerate}[({A}1)]
    \item $\mathbb S = \{\mathcal{D}_i =  (x_i, z_i, y_i) \in \mathbb R^{d_\beta} \times \mathbb R^{d_g} \times \mathbb R: i=1,2,\ldots,n \}$ is a random sample of size $n$ drawn from model \eqref{model}, with $d_\beta \leq d_g <\infty$. \label{A:sample}
    \item For some fixed $n_0 \in \mathbb N$, $\mathbb S_0 = \{ \mathcal{D}_{0j} =  (x_{0j}, z_{0j}, y_{0j}) : j=1,\ldots, n_0\}$ is an initialization random sample drawn from model \eqref{model}, set aside from $\mathbb S$. \label{A:initial}
    \item $\lambda_{\min}(G' G) > 0$ with $G = \mathbb E [G_i]= \mathbb E [z_i x_i']$, and $\lambda_{\min}(\mathbb{E} [ z_i z_i']) > 0$. \label{A:eigen}
    \item $\beta_* \in \mathbb R^ {d_\beta}$ is the (unique) solution to $\mathbb E[g_i(\beta)] = \mathbb E[ z_i \left( x_i'\beta - y_i \right)] = 0$. \label{A:iden}
   \item $\gamma_i = \gamma_0 i^ {-a}$ for some $\gamma_0 > 0$ and $a \in (1/2, 1)$. \label{A:learning}

    \item $\mathbb E \|G_i\|^ {2} < \infty$, and $\mathbb{E}[\|g_i(\beta_*)\|^2] < \infty$. \label{A:moment:consistency}
    
    
    \item $\mathbb E\|G_i\|^{2p} < \infty$ and $\mathbb E \|g_i(\beta_*)\|^{2p}<\infty$ for some integer $p > (1 - a)^{-1}$. \label{A:moment:FCLT}
    
\end{enumerate}
\end{asm*}


Assumption~\ref{A:initial} is an initialization sample that is used to construct $(\beta_0,\Phi_0,W_0)$. 
Condition~\ref{A:eigen} amounts to identification conditions and equivalent to the conditions that $G$ has full rank $d_\beta$ and $\mathbb{E} [ z_i z_i']$ is non-singular. 
Condition~\ref{A:iden} defines $\beta_*$ and its uniqueness is guaranteed by \ref{A:eigen}.
Assumption~\ref{A:learning} is the standard condition for the learning rate in the literature \citep{polyak1992acceleration}.
Conditions~\ref{A:moment:consistency}-\ref{A:moment:FCLT} impose moment conditions:
\ref{A:moment:consistency} is a less stringent assumption that ensures that the non-averaged estimator $\beta_n$ is strongly consistent for $\beta_*$ and its $L_2$ convergence rate is $O(\gamma_n)$. It is also used to obtain asymptotic normality of the averaged estimator $\bar \beta_n$, which converges faster than $\beta_n$; 
\ref{A:moment:FCLT} is a more stringent condition under which we obtain the functional central limit theorem (FCLT) for the sequence of S2SLS estimators $\beta_i$.


The following lemma establishes strong consistency of $\beta_n$.

\begin{lem}[Strong consistency]
\label{lem:strong consistency of beta_n}
Let Assumptions~\ref{A:sample}-\ref{A:iden}, \ref{A:moment:consistency} hold, and $\gamma_i=\gamma_0 i^{-a}$ for $a\in (1/2,1]$. Then, 
 as $n \to \infty$, $\beta_n \to \beta_*$ almost surely.
\end{lem}

Lemma~\ref{lem:strong consistency of beta_n} is non-trivial to prove because our proposed algorithm is based on the Z-estimator, not on the M-estimator.
The proof of Lemma~\ref{lem:strong consistency of beta_n} relies on martingale techniques: in particular, Robbins and Siegmund \citep{robbins1971convergence}, which provides  a convergence theorem for non-negative ``almost supermartingales.''\footnote{
See, e.g., Chapter 5 of \citet{benveniste2012adaptive} for an application of the Robbins-Siegmund theorem to the Robbins-Monro algorithm \citep{Robbins-Monro-1951}.} Moreover, as for the original Robbins-Monro algorithm, the almost sure convergence of $\beta_n$ allows for the learning rate of $\gamma_i =1/i$.

We now present asymptotic normality of the averaged estimator $\bar\beta_n$.

\begin{thm} 
    \label{thm:1}
      Let Assumptions~\ref{A:sample}-\ref{A:moment:consistency} hold. Then, as $n\to \infty$, 
    \begin{equation*}
        \sqrt {n}  (\bar \beta_{n} - \beta_*) \overset{d}{\to} N(0, \Avar(\bar \beta)),
    \end{equation*}    
    where $\Omega \equiv \var(g_i(\beta_*))$, $W = ( \mathbb{E} [ z_i z_i'] )^{-1}$ and $\Avar(\bar \beta) \equiv (G' W G)^{-1}G' W \Omega W G(G'W G)^{-1}$. 
\end{thm}

To prove Theorem~\ref{thm:1}, it is necessary to extend Theorems 1 and 2 of \cite{polyak1992acceleration} 
to accommodate the additional dynamics due to $\Phi_{i}$ and $W_i$. Since $\Phi_{i} \ne G$ in general, 
we must carefully consider the error $\Phi_{i} - G$. It is central to control this error to obtain asymptotic normality. 

\begin{remark}
The limiting distribution of our averaged S2SLS is the same as that of the standard (offline) 2SLS estimator. 
We will propose an efficient estimator in Section~\ref{sec:eff:est}.
\end{remark}

\begin{remark}
If $x_i$ is exogenous in model \eqref{model}, then we can take $z_i = x_i$, $G = \mathbb{E}[ x_i x_i']$, $W = (\mathbb{E}[ x_i x_i'])^{-1}$, and $\Omega = \mathbb{E}[ u_i^2 x_i x_i']$. This implies that 
$\Avar(\bar \beta) = \mathbb{E}[ x_i x_i']^{-1}  \mathbb{E}[ u_i^2 x_i x_i'] \mathbb{E}[ x_i x_i']^{-1}$, which is exactly identical to the asymptotic variance for the standard (offline) OLS estimator.
In other words, when $x_i$ is exogenous, our S2SLS estimator is not algebraically equivalent to the standard SGD-OLS estimator, but it is first-order asymptotically equivalent to it.
\end{remark}

We now strengthen Theorem \ref{thm:1} to the following functional central limit theorem (FCLT).

\begin{thm} 
\label{thm:2}
Let Assumptions~\ref{A:sample}-\ref{A:learning} and \ref{A:moment:FCLT} hold. Then: as $n\to \infty$,
    \begin{equation*}
        \left\{\frac 1 {\sqrt {n}}  \sum_{i=1}^{\lfloor n r\rfloor}(\beta_i - \beta_*) \right\}_{r\in[0,1]}\rightsquigarrow \Avar(\bar \beta)^{1/2}\{ \mathbb{W}_{d_\beta}(r)\}_{r\in [0,1]},
    \end{equation*}    
    where $\{\mathbb{W}_{d_\beta}(r)\}_{r\in [0,1]}$ denotes the $d_\beta$-dimensional standard Wiener process.
\end{thm}
The FCLT in Theorem~\ref{thm:2} states that the partial sum of the sequentially updated estimates $\beta_i$ converges weakly to a rescaled Wiener process, with the scaling matrix equal to a square root of the asymptotic variance of $\bar \beta_{n}$. 
Note that Theorem~\ref{thm:1} is a special case of Theorem~\ref{thm:2} with $r=1$ (albeit Theorem~\ref{thm:1} is derived under milder moment conditions).
Theorem~\ref{thm:2} allows us to construct robust online confidence regions for $\beta_*$; see Subsection~\ref{sec:robust-CBs} below.

\section{Efficient Estimation}\label{sec:eff:est}

\subsection{SGMM Algorithm}

In general, the S2SLS estimator in \eqref{algo:siv:2sls} is not efficient for $\beta_*$, just like the standard 2SLS estimator is inefficient.
To obtain an efficient estimator of $\beta_*$, we now propose to implement the following procedure.
First, we randomly partition the main sample into two subsamples:
$\mathbb{S} = \mathbb{S}_1 \cup \mathbb{S}_2$, where the sample size of $\mathbb{S}_j$ is denoted by $n_j$, where $j=1,2$. Thus, $n = n_1 + n_2$.
Using $\mathbb{S}_1$, we run \eqref{algo:siv:2sls} until $i=n_1$, and then using $\mathbb{S}_2$, we  sequentially update from $i = n_1 + 1$ until $i = n$: 
 \begin{subequations}\label{algo:siv:eff}
\begin{align}
    \beta_i &= \beta_{i-1} - \gamma_i (\Phi_{i-1}' W_{i-1} \Phi_{i-1})^{\dagger} \Phi_{i-1}' W_{i-1} g_i(\beta_{i-1}), \label{algo:beta:eff} \\ 
    \Phi_i &= \frac{n_0 + i - 1}{ n_0 + i} \Phi_{i-1} + \frac{1} { n_0 + i} G_i, \label{algo:Phi:eff} \\
            m_i&=n_0+i-1+g_i(\bar\beta_{n_1})'W_{i-1} g_i(\bar\beta_{n_1}), \label{algo:m:W:eff} \\  
        W_{i}&= \frac{n_0+i}{n_0+i-1} W_{i-1}  \left[I- m_i^{-1}g_i(\bar\beta_{n_1})g_i(\bar\beta_{n_1})'W_{i-1} \right], \label{algo:beta:W:eff} \\    
    \bar \beta_i &= \frac{i-1 }{ i} \bar \beta_{i-1} + \frac{1} { i}\beta_i.
    \label{algo:betabar:eff}
\end{align}
\end{subequations}

To achieve efficiency, we assume that $n_1 \rightarrow \infty$ but $n_1 / n \rightarrow 0$. 
In practice, the iterations up to $n_1$ can be viewed as a ``warm-up'' stage to avoid any too abrupt path in $\beta_i$.

\begin{remark} Note that our efficient algorithm \eqref{algo:siv:eff} is virtually the same as the inefficient algorithm \eqref{algo:siv:2sls}, except that we now update the weighting matrix $W$ differently, aiming for the optimal weighting $W=\Omega^{-1} \equiv [\var(g_i (\beta_*))]^{-1}$. As in Remark~\ref{rem:woodbury:W}, we apply SMW formula to sequentially update $W=\Omega^{-1}$ in \eqref{algo:beta:W:eff}.
Also, note that we keep the same $\bar\beta_{n_1}$ in \eqref{algo:m:W:eff} and \eqref{algo:beta:W:eff}, which is a consistent estimator for $\beta_*$.
\end{remark}

\begin{remark}[Computation of $\mathcal{V}_{i-1} =(\Phi_{i-1}' W_{i-1} \Phi_{i-1})^{\dagger}$ for efficient estimation]\label{rem:woodbury:2}
As in Remark~\ref{rem:woodbury:1}, it would be demanding to directly compute
$\mathcal{V}_{i-1} =(\Phi_{i-1}' W_{i-1} \Phi_{i-1})^{\dagger}$ in \eqref{algo:beta:eff} when $d_\beta$ is large. 
Analogous to \eqref{Woodbury:H:matrix}, we use SMW formula: 
for $\mathcal{V}_i := (\Phi_{i}'W_{i} \Phi_{i})^{-1}  \in \mathbb{R}^{d_{\beta} \times d_{\beta}}$,
\begin{equation*}
    \mathcal{V}_i = \frac {n_0 + i} {n_0 + i-1} \left( \mathcal{V}_{i-1} - \mathcal{V}_{i-1} \mathcal{U}_i (\mathcal{D}_i + \mathcal{U}_i'\mathcal{V}_{i-1}\mathcal{U}_i )^{-1} \mathcal{U}_i' \mathcal{V}_{i-1}\right),    
\end{equation*}
where 
\begin{align*}
    \mathcal{D}_i &=\operatorname{diag}\left[-z_i'W_{i-1} z_i, z_i'W_{i-1} z_i, -m_i\right] \in \mathbb{R}^{3 \times 3},\\
    \mathcal{U}_i &= \left[\Phi_{i-1}' W_{i-1} z_i, \Phi_{i-1}' W_{i-1} z_i + \frac{z_i'W_{i-1}z_i}{n_0 + i-1}x_i, \Phi_{i-1}' W_{i-1} g_i(\bar\beta_{n_1}) + \frac{z_i'W_{i-1}g_i(\bar\beta_{n_1})}{n_0 + i-1}x_i\right] \in \mathbb{R}^{d_{\beta} \times 3}.
\end{align*}
Since $\mathcal{D}_i + \mathcal{U}_i' \mathcal{V}_{i-1} \mathcal{U}_i$ is a 3 by 3 matrix, inverting it would require much less computation compared to directly inverting $\Phi_{i-1}' W_{i-1}\Phi_{i-1}$ when $d_\beta \gg 3$.
\end{remark}

\subsection{Asymptotic Efficiency}

We make the following additional regularity condition.
\begin{asm*}\label{asm:conditions:eff}
\begin{enumerate}[({A}8)]
    \item $n_1 \to \infty$, $n_1/n \to 0$, $\mathbb{E}[\|\beta_0\|^{2\tilde{p}} \,] < \infty$ for some constant $\tilde{p}$, and
 $\inf_{\beta \in \mathcal{K}} \lambda_{\min}(\mathbb{E}[g_i(\beta)g_i(\beta)']) > 0$ for some compact set $\mathcal{K}$ that contains $\beta_*$ in its interior. \label{A:eff:8}
    
\end{enumerate}
\end{asm*}

The following theorems establish asymptotic properties of SGMM. 

\begin{thm}
\label{thm:3} 
Let Assumptions \ref{A:sample} -- \ref{A:moment:consistency} and \ref{A:eff:8} hold with $\tilde{p}=1$.
Then, as $n \to \infty$, $\beta_n$ and $\bar \beta_n$ are weakly consistent for $\beta_*$ and
\begin{equation*}
    \sqrt{n}(\bar \beta_n - \beta_*) \overset{d}{\to} N(0, (G' \Omega^{-1} G)^{-1}).
\end{equation*}
\end{thm}

Theorem~\ref{thm:3} shows that SGMM is asymptotically first-order equivalent to the standard (offline) efficient GMM estimator. Theorem~\ref{thm:4} below strengthens Theorem~\ref{thm:3} by establishing the FCLT under extra moment condition. 

\begin{thm}
\label{thm:4} 
Let Assumptions \ref{A:sample} - \ref{A:learning}, \ref{A:moment:FCLT} and \ref{A:eff:8} hold with $\tilde{p}=p$, where $p$ is defined in Assumption \ref{A:moment:FCLT}.
Then, it holds
\begin{equation*}
    \left\{ \frac 1 {\sqrt n}\sum_{i=1}^{\lfloor nr \rfloor} (\beta_i - \beta_*) \right\}_{r\in[0,1]}\rightsquigarrow (G'\Omega^{-1}G)^{-1/2} \{ \mathbb{W}_{d_\beta}(r) \}_{r \in [0, 1]}.
\end{equation*}
\end{thm}

Theorem~\ref{thm:4} suggests robust online confidence regions as those presented in Subsection~\ref{sec:robust-CBs} below.

\section{Inference}\label{sec:inference}

In this section, we first present two simple methods to construct fast online confidence regions. We then show that a couple of well-known statistical tests can be seamlessly integrated within the SGMM framework.\footnote{The purpose of this section is to showcase the usefulness of our approach. It is a topic for future research to investigate a variety of inference problems more extensively.}

\subsection{Online Confidence Regions}\label{sec:robust-CBs}

We propose two simple online confidence regions: the plug-in (PI) base approach and the random scaling (RS) approach.

\subsubsection{Plug-in consistent online confidence regions}
%
As a by-product of the efficient algorithm, $W_n$ defined in \eqref{algo:beta:W:eff} consistently estimates $\Omega^{-1}=(\var[g_i(\beta_*)])^{-1}$, and the $\mathcal{V}_n=(\Phi_{n}' W_{n} \Phi_{n})^{\dagger}$ defined in Remark~\ref{rem:woodbury:2} consistently estimate the asymptotic efficient variance $(G' \Omega^{-1} G)^{-1}$ in Theorem~\ref{thm:3}. Hence, we can conduct asymptotic optimal inference using $\mathcal{V}_n =(\Phi_{n}' W_{n} \Phi_{n})^{\dagger}$, and the resulting inference will be called ``plug-in inference''.  In particular, we can bulid the optimal plug-in online confidence regions based on
\begin{equation}\label{Wald-eff}
\mathcal{W}_{eff}:=n(\bar{\beta}_n - \beta_*)'(\Phi_{n}' W_{n} \Phi_{n})(\bar{\beta}_n - \beta_*)
\overset{d}{\to}\mathcal{X}_{d_{\beta}}.
\end{equation}

\subsubsection{Random scaling robust online confidence regions}
For both inefficient S2SLS and efficient SGMM, 
given the FCLT Theorems~\ref{thm:2} and ~\ref{thm:4}, we can apply the random scaling method proposed in 
\cite{lee2021fast, lee2022fastqr}, which is based on the following robust, inconsistent long-run variance (LRV) estimate idea of \citet{kiefer2000simple,velasco2001edgeworth,gupta2021robust} for $\Avar(\bar \beta)$: 
\begin{equation}\label{HAC-kernel}
\widehat V_{rs,n}:=\frac{1}{n}\sum_{s=1}^n\left(\frac{1}{\sqrt{n}}\sum_{i=1}^s(\beta_i-\bar\beta_n)\right)\left(\frac{1}{\sqrt{n}}\sum_{i=1}^s(\beta_i-\bar\beta_n)\right)'.
\end{equation}
See \cite{lee2021fast} for simple online version to compute $\widehat V_{rs,n}$ sequentially.
Then a robust online confidence region for $\beta_{*}$ can be constructed using the following statistic:
\begin{align}\label{Wald-rs}
\mathcal W_{rs}  &  \equiv\frac{n}{d_{\beta}}(\bar
{\beta}_{n}-\beta_{*})^{\prime}\widehat{V}_{rs,n}^{-1}(\bar{\beta}%
_{n}-\beta_{*})\nonumber\\
&  \overset{d}{\to}\frac{1}{d_{\beta}}\mathbb{W}_{d_{\beta}}(1)'
\left(\int_0^1[\mathbb{W}_{d_{\beta}}(r) -r \mathbb{W}_{d_{\beta}}(1)]
[\mathbb{W}_{d_{\beta}}(r) -r \mathbb{W}_{d_{\beta}}(1)]'dr\right)^{-1}\mathbb{W}_{d_{\beta}}(1),
\end{align}
where the critical values can be simulated as in \citet{kiefer2000simple}.

We have implemented online confidence sets using $\mathcal W_{eff}$ and $\mathcal W_{rs}$ versions in Monte Carlo experiments and real data applications below. 

\subsection{Online Endogeneity Tests}\label{sec:DWH-test}

To further illustrate the usefulness of our inference method, we now consider an endogeneity test focusing on only a subset  $\beta_{\text{sub}}$ of $\beta_*$. Under the null, the probability limits of OLS and IV estimators are the same; under the alternative, the IV estimator is still consistent but OLS is not. 
  Let $\alpha_{\sub}$ denote the probability limit of OLS for the subvector. The null hypothesis is then 
\begin{align}\label{null:DWH}
    \mathbb H_0:  \alpha_{\sub}= \beta_{\sub}.
\end{align}

We propose an online algorithm to implement the  Durbin-Wu-Hausman (DWH) test. Let $\beta_i$ and $\alpha_i$ respectively denote the stochastic sequences of the IV-estimator and OLS. They are jointly updated as follows:
\begin{equation}
\begin{pmatrix}
\beta_i\\
\alpha_i
\end{pmatrix} =\begin{pmatrix}
\beta_{i-1}\\
\alpha_{i-1}
\end{pmatrix}   - \gamma_i\begin{pmatrix}
(\Phi_{i-1}' W_{i-1} \Phi_{i-1})^{\dagger} \Phi_{i-1}' W_{i-1} z_i (x_i'\beta_{i-1}-y_i)\\
x_i (x_i'\alpha_{i-1}-y_i)
\end{pmatrix}.
\end{equation}
\begin{remark}
 Note that $\alpha_i$ follows the usual SGD path for M-estimation. To improve the finite-sample performance, we could have multiplied $x_i(x_i'\alpha_{i-1}-y_i)$ by $(\frac{1}{n_0+i-1}\sum_{j\le n_0} x_{0j}x_{0j}' + \frac{1}{n_0+i-1}\sum_{j\le i-1} x_j x_j')^{-1}$.
\end{remark}
Let $\bbeta_i=(\beta_i',\alpha_i')'$ denote the vector stacking all elements of the updating sequences and let $\bar\bbeta_n=\frac{1}{n}\sum_{i=1}^n\bbeta_i$. Let $$\bbeta_{*}=
\begin{pmatrix}
\beta_*\\
(\mathbb Ex_ix_i')^{-1}\mathbb E x_i y_i
\end{pmatrix}.
$$
We show that under either the null or the alternative, for some covariance matrix $\Gamma$, 
 \begin{equation*}
        \left\{\frac 1 {\sqrt {n}}  \sum_{i=1}^{\lfloor n r\rfloor}(\bbeta_i - \bbeta_*) \right\}_{r\in[0,1]}\rightsquigarrow \Gamma^{1/2}\{ \mathbb{W}_{2d_\beta}(r)\}_{r\in [0,1]},
    \end{equation*}    
The above FCLT allows us to construct a simple online DWH test for endogeneity.
The test statistic has to be properly scaled using the asymptotic variance. While the  scaling asymptotic variance in the DWH test stems from the idea of estimation efficiency, it is computationally demanding to implement its exact form via stochastic approximation in the online context. We adopt the above random scaling robust LRV estimate $\widehat{V}_{rs,n}$ in \eqref{HAC-kernel} instead. In particular we use
$$
\widehat{V}_{n}:=\frac{1}{n}\sum_{s=1}^n\left(\frac{1}{\sqrt{n}}\sum_{i=1}^s(\bbeta_i-\bar\bbeta_n)\right)\left(\frac{1}{\sqrt{n}}\sum_{i=1}^s(\bbeta_i-\bar\bbeta_n)\right)'.
$$
See \cite{lee2021fast} for updating ($\bar\bbeta_n,~\widehat V_{n}$) sequentially.

Let $\bbeta_{\sub,i}$ denote the subvector of $\bbeta_i$, corresponding to $\beta_{\sub}$ and $\alpha_{\sub}$.  In the algorithm, $\bbeta_i$, $\Phi_i$ and $W_i$ are potentially high-dimensional objects. Instead of sequentially update the full vector/matrix $\bar\bbeta_i$ and $\widehat{V}_{i}$, we just need to update the subvector $\bar\bbeta_{\sub,i}$ and its corresponding submatrix $\widehat{V}_{\sub,i}$.

Note that we can express $\bar\bbeta_{\sub,n}=(\bar\beta_{\sub,n}',\bar\alpha_{\sub,n}')'$ corresponding to the online IV and OLS estimators. The online DWH test is then conducted by comparing $\bar\beta_{\sub,n}$ and $\bar\alpha_{\sub, n}$, which can be expressed as
$$
\bar\beta_{\sub,n}-\bar\alpha_{\sub, n}= \Xi \bar\bbeta_{\sub,n},\text{ where } \Xi =(I, -I).
$$
Let $q$ denote the number of restrictions in the null hypothesis \eqref{null:DWH}.
The pivotal statistic is now defined as 
$$
\mathcal{S}_{rs}:= \frac{n}{q}(\bar\beta_{\sub,n}-\bar\alpha_{\sub, n})'(\Xi\widehat{V}_{\sub,n}\Xi')^{-1} (\bar\beta_{\sub,n}-\bar\alpha_{\sub, n}).
$$
The asymptotic distribution of the pivotal statistic can be derived using the FCLT of the stacked vector. This implies the asymptotic null distribution of the pivotal statistic, stated as follows. 

\begin{cor}\label{cor:dwh}
Let assumptions for Theorem~\ref{thm:2} and Theorem~\ref{thm:4} hold, respectively. Under the null hypothesis that $\alpha_{\sub}=\beta_{\sub}$, we have:
$$
\mathcal{S}_{rs}  \rightarrow\!_{d} ~\frac{1}{q} \mathbb{W}_{q}(1)'\left(\int_0^1[\mathbb{W}_{q}(r) -r \mathbb{W}_{q}(1)][\mathbb{W}_{q}(r) -r \mathbb{W}_{q}(1)]'dr\right)^{-1}\mathbb{W}_{q}(1),
$$
where $q = \mathrm{dim}(\alpha_{\sub})=\mathrm{dim}(\beta_{\sub})$.
\end{cor}
Critical values for testing linear restrictions are given in \citet[Table II]{kiefer2000simple}. 
%

\subsection{Online Sargan-Hansen Tests}\label{Jtest}

As a straightforward corollary to the main result in Section \ref{sec:eff:est}, which proposes an online efficient GMM estimation, we also implement the test for the  overidentifying moment restrictions. Let $\hat{g}_{n_{1}}=\frac{1}{n_{1}}\sum_{i=1}^{n_{1}}g_{i}\left(\bar{\beta}_{n_{1}}\right)$
and add the following to the end of the efficient algorithm \eqref{algo:siv:eff}: from $i = n_1 + 1$ until $i = n$, 
\begin{eqnarray*}
\hat{g}_{i} & = & \frac{i-1}{i}\hat{g}_{i-1}+\frac{1}{i}g_{i}\left(\bar{\beta}_{i}\right).
\end{eqnarray*}
Then, we obtain the conventional chi-squared test for the overidentifying restrictions. 

\begin{cor}\label{overid:test}
Let Assumptions in Theorem~\ref{thm:3} hold with $d_g>d_{\beta}$. Then we have:
\[
\mathcal {J}_{n}:=n\hat{g}_{n}'W_{n}\hat{g}_{n}\overset{d}{\to}\mathcal{X}_{d_{g}-d_{\beta}},
\]
as $n\to\infty$, where $W_n$ is defined in \eqref{algo:beta:W:eff}.
\end{cor}


\section{Monte Carlo Experiments}\label{sec:Monte-Carlo}

In this section, we investigate the numerical performance of the SGMM estimator via Monte Carlo experiments. Initially, we discuss the process of selecting the learning rate $\gamma_1$, which will be useful when working with a real data set.

\subsection{Selection of the Learning Rate in Applications}\label{subsec:Selection_gamma0}

In this section, we describe a rule of thumb regarding how to choose $\gamma_i = \gamma_0 i^ {-a}$.
Suppose that $a \in (1/2,1)$ is fixed at a given constant (in the examples reported below, we set $a=0.501$). Then, it remains to choose $\gamma_0 > 0$, that is, 
the initial value of the learning rate.
Recall that we have the initialization sample $\mathbb S_0$ with sample size $n_0$ to compute 
$\beta_0$,  $\Phi_0$ and $W_0$ and start with the first update as   
\begin{align}\label{sgmmm:beta1}
    \beta_1 &= \beta_{0} - \gamma_1 (\Phi_{0}' W_{0} \Phi_{0})^{\dagger} \Phi_{0}' W_{0} (G_1\beta_0 + H_1), 
\end{align}
where $\gamma_1 = \gamma_0$.
We first define
\begin{align}\label{sgmmm:gamma0:def:Psi}
\Psi_0(\alpha) := \text{quantile}_{1-\alpha} \{ d_{\beta}^{-1}  \| (\Phi_{0}' W_{0} \Phi_{0})^{\dagger} \Phi_{0}' W_{0} G_{0i} \|_2: G_{0i} \in \mathbb S_0;\;  i =1, \ldots, n_0 \},
\end{align}
where $\| \cdot \|_2$ is the spectral norm. 
We propose to use 
\begin{align}\label{sgmmm:gamma0:def}
   \gamma_0 = \frac{1}{\Psi_0(\alpha)}, 
\end{align}
where $\alpha$ is a predetermined quantile level (e.g., $\alpha = 0.5$).
The rational behind this rule of thumb is that we choose $\gamma_0$ small enough 
such that it is likely that the $\beta_i$ path is not explosive when $i$ is relatively small.  


\subsection{Simultation Results}
We consider the following data generating process as a baseline model:
\begin{align}\label{eq:sim_reg}
    y_i = x_i'\beta_* + \varepsilon_i ~~~\mbox{ for } i=1,\ldots,n, 
\end{align}
where $x_i$ is a $p$-dimensional vector of regressors, with the first element being endogenous. There exists a $q$-dimensional vector of exogenous variables $z_i$, which follows a multivariate normal distribution $N(0,\Sigma)$. The $(i,j)$ element of $\Sigma$ is set to be $\Sigma_{i,j}:=\rho^{|i-j|}$. The endogenous regressor $x_{i,1}$ is generated as follows: for some $\underline{p} \le p$ and $\underline{q} \le q$,
\begin{align}
    x_{i,1} = 0.1 \times \sum_{j=2}^{\underline{p}} x_{i,j} + 0.5 \times \sum_{j=\underline{p}}^{\underline{q}} z_{i,j} + \nu_i,
\end{align}
where $x_{i,j} = z_{i,j-1}$ for $j=2,\ldots,\underline{p}$ and $\nu_i \sim N(0,1)$. Finally, the error term in \eqref{eq:sim_reg} is generated by
\begin{align}
    \varepsilon_i = \sigma_i \cdot (\nu_i + \eta_i),
\end{align}
where $\sigma_i = 5 \cdot \exp(z_{i,\underline{q}})$ and $\eta_i\sim N(0,1)$. Therefore, the model allows for both heteroskedasticity and endogeneity. 

We consider four different sample sizes $n=\{10^4, 10^5, 10^6,10^7\}$. We set the correlation coefficient of $z_i$ as $\rho=0.5$ and the true regression coefficients as $\beta_*=(1,\ldots,1)$. The dimensions of $x_i$ and $z_i$ are set to $(p,q)=(5,20)$ and $(10,25)$ with $(\underline{p},\underline{q})= (5,20)$. Therefore, we conduct the Monte Carlo experiments over 8 different designs. We replicate each design 1,000 times to compute the performance statistics. 

The simulations are conducted using the Graham cluster of the Digital Research Alliance of Canada, which consists of several Intel CPUs (Broadwell, Skylake, and Cascade Lake) operating at frequencies between 2.1GHz and 2.5GHz. The memory budget is set to 64 gigabytes of RAM.

Tables \ref{tb: (p,q)=(5,20)}--\ref{tb: (p,q)=(10,25)} summarize the simulation results. We estimate the model using two different weight schemes, as described in Sections \ref{sec:baseline-weight} and \ref{sec:eff:est}. We denote them as S2SLS and SGMM, respectively. To compare the performance, we also estimate the model using the offline counterparts: 2SLS and GMM through R packages \texttt{ivreg} (CRAN version 0.6.2) and \texttt{gmm} (CRAN version 1.7.0), respectively. 

For S2SLS and SGMM, we need to set some tuning parameters and initial values. The learning rate $\gamma_i \equiv \gamma_0 i^{-a}$ is set with $a=0.501$ and $\gamma_0$ as the rule of thumb method described in section \ref{subsec:Selection_gamma0}. This size of an initialization sample is set to $n_0=1000$. Using the initialization sample, we estimate the initial value $\beta_0$ by 2SLS, $\Phi_0=n_0^{-1}\sum_{j=1}^{n_0} G_{0j}$, and $W_0 = (n_0^{-1}\sum_{j=1}^{n_0} z_{0j}z_{0j}')^{-1}$. Finally, we fix $n_1=10\sqrt{n}$ for SGMM. Two alternative methods for inference are considered: SGMM RS
(specifically, as in $\mathcal W_{rs}$ in \eqref{Wald-rs})
and SGMM PI, respectively, refer to random scaling (RS) and plug-in (PI) inference with the same point estimate SGMM.

In the tables, we focus on the coefficient of the endogenous regressor $x_{i,1}$ and report the following performance statistics: root mean square error (RMSE), average bias (Bias), standard deviation (SD), coverage probability of the 95\% confidence interval (Coverage Prob), the average confidence interval length (CI Length), and the average computation time in seconds (Time). 

\begin{table}[th]
\caption{Simulation Results with $(p,q)= (5,20)$} \label{tb: (p,q)=(5,20)}
\begin{threeparttable}
\begin{tabular}{l c c c c c S[table-format=3.2]}
\toprule
& RMSE & Bias & SD & Coverage Prob & CI Length & \mbox{Time (sec.)}\\
\midrule
\underline{$n=10^4$} \\
2SLS & 0.06833 & 0.00165 & 0.06831 & 0.945 & 0.26464 & 0.08 \\
GMM & 0.05858 & 0.00150 & 0.05856 & 0.942 & 0.22480 & 2.84 \\
S2SLS & 0.07004 & 0.00109 & 0.07003 & 0.955 & 0.34386 & 0.09 \\
SGMM RS & 0.06946 & 0.00314 & 0.06939 & 0.954 & 0.34677 & 0.20 \\
SGMM PI & 0.06946 & 0.00314 & 0.06939 & 0.875 & 0.20418 & 0.19 \\
\underline{$n=10^5$} \\
2SLS & 0.02058 & -0.00002 & 0.02058 & 0.963 & 0.08491 & 0.96 \\
GMM & 0.01805 & -0.00052 & 0.01804 & 0.957 & 0.07479 & 22.96 \\
S2SLS & 0.02092 & -0.00018 & 0.02092 & 0.955 & 0.11270 & 0.82 \\
SGMM RS & 0.01896 & -0.00064 & 0.01895 & 0.958 & 0.10337 & 1.93 \\
SGMM PI & 0.01896 & -0.00064 & 0.01895 & 0.940 & 0.07334 & 1.90 \\
\underline{$n=10^6$} \\
2SLS & 0.00706 & -0.00004 & 0.00706 & 0.943 & 0.02693 & 13.44 \\
GMM & 0.00625 & -0.00022 & 0.00625 & 0.935 & 0.02386 & 263.67 \\
S2SLS & 0.00706 & -0.00006 & 0.00706 & 0.942 & 0.03511 & 8.17 \\
SGMM RS & 0.00630 & -0.00023 & 0.00630 & 0.950 & 0.03163 & 19.96 \\
SGMM PI & 0.00630 & -0.00023 & 0.00630 & 0.934 & 0.02374 & 19.65 \\
\underline{$n=10^7$}\\
2SLS & 0.00223 & 0.00009 & 0.00223 & 0.943 & 0.00851    & 150.33 \\
GMM & NA & NA & NA & NA & NA                          & NA \\
S2SLS & 0.00223 & 0.00008 & 0.00223 & 0.941 & 0.01110   & 77.23 \\
SGMM RS & 0.00199 & 0.00009 & 0.00199 & 0.937 & 0.00977 & 193.20 \\
SGMM PI & 0.00199 & 0.00009 & 0.00199 & 0.935 & 0.00754 & 189.89 \\
\bottomrule
\end{tabular}
\begin{tablenotes}
\footnotesize
\item \emph{Notes.} These results are based on 1,000 replications. `RMSE', `Bias', and `SD' are obtained over simulation draws. `Coverage Prob' denotes coverage probability computed for the 95\% confidence interval. `CI Length' denotes the average length of the confidence interval. GMM does not meet the memory budget of 64 gigabytes when $n=10^7$ and is denoted as `NA (Not Available)'. The average computation time is measure in seconds.  
\end{tablenotes}
\end{threeparttable}
\end{table}

\begin{table}[th]
\caption{Simulation Results with $(p,q)= (10,25)$} \label{tb: (p,q)=(10,25)}
\begin{tabular}{l c c c c c S[table-format=3.2]}
\toprule
& RMSE & Bias & SD & Coverage Prob & CI Length & \mbox{Time (sec.)}\\
\midrule
\underline{$n=10^4$} \\
2SLS & 0.09251 & 0.00355 & 0.09244 & 0.947 & 0.35381 & 0.12 \\
GMM & 0.07579 & 0.00198 & 0.07577 & 0.951 & 0.28601 & 4.09 \\
S2SLS & 0.13041 & -0.00274 & 0.13038 & 0.952 & 0.61009 & 0.16 \\
SGMM RS & 0.12553 & -0.00171 & 0.12552 & 0.966 & 0.58261 & 0.44 \\
SGMM PI & 0.12553 & -0.00171 & 0.12552 & 0.833 & 0.27651 & 0.44 \\
\underline{$n=10^5$}\\
2SLS & 0.02870 & -0.00011 & 0.02870 & 0.955 & 0.11265 & 1.40 \\
GMM & 0.02476 & 0.00039 & 0.02476 & 0.951 & 0.09463 & 29.59 \\
S2SLS & 0.02989 & -0.00056 & 0.02989 & 0.945 & 0.15533 & 1.35 \\
SGMM RS & 0.02704 & 0.00039 & 0.02703 & 0.948 & 0.13795 & 3.99 \\
SGMM PI & 0.02704 & 0.00039 & 0.02703 & 0.917 & 0.09338 & 3.93 \\
\underline{$n=10^6$}\\
2SLS & 0.00907 & 0.00048 & 0.00906 & 0.953 & 0.03568 & 25.94 \\
GMM & 0.00755 & 0.00022 & 0.00755 & 0.954 & 0.03021 & 305.57 \\
S2SLS & 0.00908 & 0.00045 & 0.00907 & 0.959 & 0.04774 & 13.50 \\
SGMM RS & 0.00762 & 0.00023 & 0.00762 & 0.959 & 0.04016 & 40.39 \\
SGMM PI & 0.00762 & 0.00023 & 0.00762 & 0.949 & 0.03009 & 39.86 \\
\underline{$n=10^7$}\\
2SLS & 0.00297 & -0.00005 & 0.00297 & 0.949 & 0.01129    & 247.53 \\
GMM & NA & NA & NA & NA & NA                           & NA \\
S2SLS & 0.00298 & -0.00005 & 0.00298 & 0.945 & 0.01479   & 129.14 \\
SGMM RS & 0.00249 & -0.00007 & 0.00249 & 0.944 & 0.01237 & 393.74 \\
SGMM PI & 0.00249 & -0.00007 & 0.00249 & 0.945 & 0.00955 & 387.85 \\
\bottomrule
\end{tabular}
\end{table}

Overall, the numerical performance of S2SLS and SGMM is satisfactory. First, both S2SLS and SGMM demonstrate good coverage probabilities across all designs. Additionally, other measures such as RMSE, Bias and SD also indicate good performance. When we examine RMSEs specifically, they are slightly larger than those of the offline estimators when $n=10^4$ and $10^5$. However, for sample sizes $n\ge10^6$, RMSEs become comparable to those of the offline estimators, aligning with the asymptotic theory in the previous sections. 

Second, both S2SLS and SGMM shows substantial gains in computation time as the sample size increases. In the model of $(p,q)=(5,20)$, 2SLS takes 1.65 times longer computation time than S2SLS, and GMM does 7.6 times more than SGMM when $n$ is bigger than $10^6$. We observe a similar pattern in the model of $(p,q)=(10,25)$. Note that we compute the whole matrix $\widehat{V}_n$ in these simulations. If we are interested in the inference of a single parameter, we can improve the result further by focusing on a single element of $\widehat{V}_n$. 

Third, SGMM demonstrates efficiency gains over S2SLS across all designs as predicted by asymptotic theory. As we discuss earlier, RMSEs of SGMM are comparable to those of GMM when the sample size is $10^6$. If the sample size is as large as $10^7$, GMM exceeds the memory budget of 64 gigabytes, resulting in a value of `NA' in the tables. These results highlight the computational advantage of SGMM over GMM while maintaining its efficiency property. 

Finally, we observe that the average confidence interval length (CI Length) is larger for S2SLS and SGMM RS than their offline counterparts. We observe similar phenomena in linear mean and quantile regression models (see \citet{lee2021fast, lee2022fastqr}.

\section{Empirical Examples}\label{sec:examples}
In this section, we explore two empirical applications to demonstrate the effectiveness of the SGMM estimator. Specifically, we revisit the empirical findings presented in \citet{angrist1991does} and \citet{angrist1998children}. 

\subsection{\citet{angrist1991does}} 

We re-visit the 2SLS estimate of return to education in column (2) of Table IV in \citet{angrist1991does}:
\begin{align*}
\log(wage_i) = \beta_0 + \beta_1 educ_i + x_i'\beta_3 + \varepsilon_i,
\end{align*}
where $wage_i$ denotes a weekly wage, $educ_i$ denotes the years of education, and $x_i$ is a vector of 9 cohort dummies. The object of interest is $\beta_1$ representing returns to schooling. A vector of instruments, $z_i$, is constructed by the interaction of quarter of birth and cohort dummies, where $d_z=30$. The model is overidentified in this application, as we have only 11 regressors. 

Table \ref{tb:AK} provides a summary of the estimation results. Similar to the simulation studies, we employ five different estimators: 2SLS, GMM, S2SLS, SGMM RS, and SGMM PI. Among the total 247,199 observations, we allocate $n_0=20,000$ into the initialization sample, resulting in $n=227,199$.\footnote{Because of exclusion of the initialization sample, the 2SLS estimate in Table~\ref{tb:AK} is slightly different from one reported in column (2) of Table IV in \citet{angrist1991does}. The latter is $0.0769$ with the standard error of $0.0150$.} Similar to the simulations studies, we set $n_1=10\sqrt{n}$ for SGMM. 
Finally, we adopt the method described in Subsection \ref{subsec:Selection_gamma0} and set $\gamma_0=0.200$. In the table, we present the point estimates of $\beta_1$, along with their corresponding 95\% confidence intervals, the lengths of these confidence intervals, and the computation times. 

    \begin{table}[hptb]
        \caption{Estimation Results of \citet{angrist1991does}} \label{tb:AK}
        \centering
        \begin{tabular}{l c c c S[table-format=3.2]}
        \hline
                & Estimate of $\hat{\beta_1}$ & 95\% CI & CI Length & {Time (sec.)} \\
        \hline
        2SLS       & 0.0764 & (0.0459, 0.1070) & 0.0611  &   5.61 \\
        GMM        & 0.0755 & (0.0450, 0.1060) & 0.0610  & 171.06 \\
        S2SLS      & 0.1108 & (0.0593, 0.1623) & 0.1030  &   5.09 \\
        SGMM RS    & 0.1113 & (0.0601, 0.1625) & 0.1024  &  13.54 \\
        SGMM PI    & 0.1113 & (0.0758, 0.1468) & 0.0710  &  13.41 \\
        SGMM ME    & 0.0790 & (0.0474, 0.1106) & 0.0632  & 132.44 \\
        \hline
        \end{tabular}
    \end{table}

We observe that both SGMM estimators are computed nearly 12.6 times faster than the corresponding offline GMM estimator, and S2SLS performs slightly than 2SLS in this application. Additionally, the confidence intervals of S2SLS and SGMM are wider than those of their offline counterparts, as we confirmed in the simulation studies. 

To explore the gap between the point estimates, we next consider implementing a multi-epoch algorithm within this application. Note that the point estimates of the stochastic methods (S2SLS and SGMM) are around 0.111, while the offline (2SLS and GMM) estimates are around 0.076. In order to assure a stable estimation result, we embark on a multi-epoch approach for the stochastic estimations.\footnote{In machine learning, 
an epoch refers to a complete pass through a training dataset in an iterative optimization algorithm. In our setting, an epoch corresponds to the use of the all observations to run S2SLS or SGMM.
In practice, training a machine learning model typically involves running through multiple epochs. This is because a single pass through the training data might not be sufficient for the estimate to converge to an optimal or near-optimal value. This seems the case with the \citet{angrist1991does} example.}
Precisely, we shuffle the order of observations within the sample during each epoch and subsequently compute S2SLS or SGMM over a series of epochs.\footnote{In other words, within each epoch, we implement uniform sampling without replacement with sample size $n$.}

Figure \ref{fig:multi-epoch} demonstrates the estimation path of SGMM PI over 10 epochs along with the pointwise 95\% confidence band. The pointwise SGMM plug-in confidence interval was calculated using a sample size of $\min\{i,n\}$, i.e.\ the variance was initially divided by $\sqrt{i}$ for $i=1,\ldots, n$, in the first epoch, and subsequently by $\sqrt{n}$ thereafter. In the graph, we can observe that the estimation process stabilizes after the fifth epoch. The estimation results after the 10th epoch are presented at the bottom of Table \ref{tb:AK}, designated as SGMM ME. The point estimate stands at 0.0790, closely aligned with the offline GMM estimate, and the length of the confidence interval is considerably reduced (0.632). Furthermore, SGMM ME necessitates only approximately half the computation time of GMM. 
Therefore, when the dataset size allows for the storage of all observations, as is the case in both examples, we recommend opting for a multi-epoch algorithm and an assessment of SGMM stability in empirical applications.\footnote{However, it is an interesting open question for future research to formally extend our theory to a multi-epoch (or multi-pass) setting and develop an early stopping rule for the number of epochs.} 

\begin{figure}[hpbt]
\caption{Estimation Path over Multiple Epochs in  \citet{angrist1991does}}\label{fig:multi-epoch}
    \includegraphics[scale=0.5]{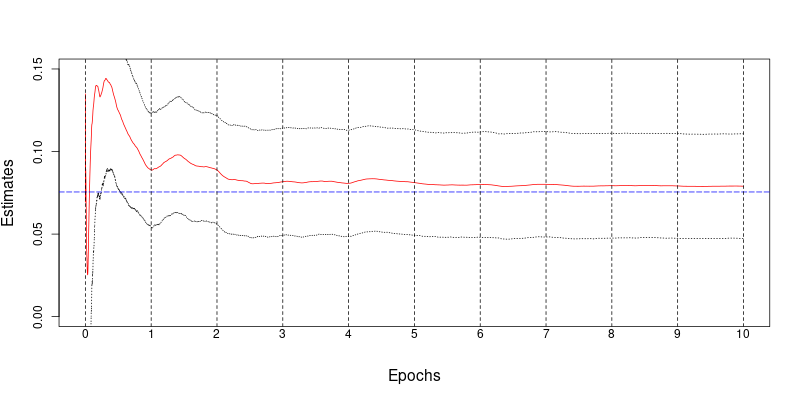}   
    \caption*{
    \footnotesize
    \emph{Notes.} The solid line denotes the estimation path of SGMM PI over 10 epochs. The dotted line around the solid line represents the path of the corresponding 95\% SGMM PI confidence intervals. The dashed horizontal line denotes the offline GMM estimate.}
\end{figure}

\subsection{\citet{angrist1998children}}

In \citet{angrist1998children}, they study the effect of childbearing on female labor supply.
In our application, we use data consisting of 394,840 observations from the 1980 U.S. census. The dependent variable
 is the number of working weeks divided by 52;
the endogenous regressor is a binary variable that takes value 1 if the number of children is greater than 2;
the instrument is a binary variable that takes value 1 if siblings are of the same sex.
To be consistent between two applications, we repeat the same exercises as in the previous subsection.
Specifically, we set $n_0=20,000$, $n = 374,840$, and $n_1=10\sqrt{n}$. The initial-data-dependent choice $\gamma_0$ was that $\gamma_0 = 0.058$. 

Table~\ref{tb:AE} and Figure~\ref{fig:multi-epoch2} present the estimation results. 
As this is a just-identified case, we expect little difference across S2SLS and SGMM, which was empirically verified.
It is interesting to notice that the SGMM estimate here basically converges only after one or two epochs, unlike the previous application.
This is likely due to the fact that the number of parameters is just two, including the intercept term, and the model is just-identified in the second example.

    \begin{table}[hptb]
        \caption{Estimation Results of \citet{angrist1998children}} \label{tb:AE}
        \centering
        \begin{tabular}{l c c c S[table-format=3.2]}
        \hline
                & Estimate of $\hat{\beta_1}$ & 95\% CI & CI Length & {Time (sec.)} \\
        \hline
        2SLS     & -0.1256  &(-0.1714,  -0.0797) &    0.0917 &  1.77     \\
        GMM      & -0.1256  &(-0.1714,  -0.0797) &    0.0917 &  1.92     \\
        S2SLS    & -0.1244  &(-0.1921,  -0.0567) &    0.1354 &  0.31     \\
        SGMM RS  & -0.1244  &(-0.1921,  -0.0567) &    0.1354 &  1.45     \\
        SGMM PI  & -0.1244  &(-0.1770,  -0.0717) &    0.1053 &  1.40     \\
        SGMM ME  & -0.1277  &(-0.1759,  -0.0796) &    0.0963 &  4.23     \\
        \hline
        \end{tabular}
    \end{table}

\begin{figure}[hpbt]
\caption{Estimation Path over Multiple Epochs in \citet{angrist1998children}}\label{fig:multi-epoch2}
    \includegraphics[scale=0.5]{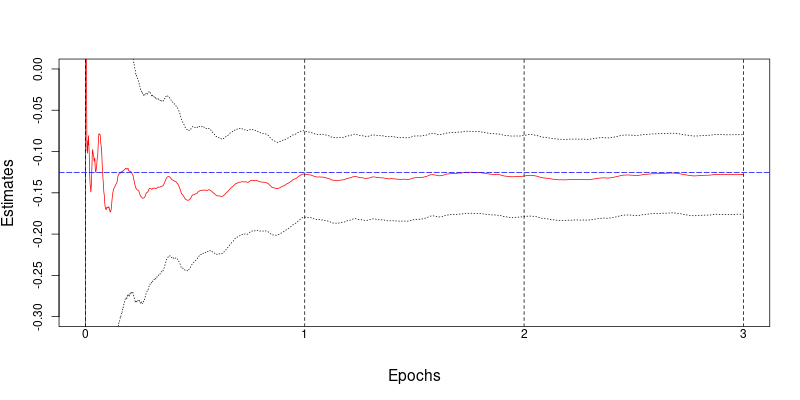}   
    \caption*{
    \footnotesize
    \emph{Notes.} Refer to the captions in Figure \ref{fig:multi-epoch} for the corresponding legends.}
\end{figure}

\section{Extensions}\label{sec:sgmm:extension}

We conclude the paper by mentioning two possible extensions. 
First, recall the standard (nonlinear) GMM estimator for the general GMM model, $\mathbb{E}[g_i (\beta_*)]=0$ with $\dim(g_i)\ge d_{\beta}$:
\begin{align}\label{def:GMM:nl}
    \min_{\beta}\bar{g}_n (\beta)' W_n \bar{g}_n (\beta),
\end{align}
where $\bar{g}_n (\beta) = n^{-1} \sum_{i=1}^n g_i(\beta)$, and $W_n$ is a weighting matrix that may depend on an initial estimator of $\beta_*$.
The first-order condition to \eqref{def:GMM:nl} is
\begin{align*}
    \frac{\partial}{\partial \beta} \bar{g}_n (\beta)' W_n \bar{g}_n (\beta) = 0.
\end{align*}
The following is a natural efficient online algorithm for nonlinear GMM.
We assume that the parameter space for $\beta$ is bounded in this case. For simplicity, we drop the step using the subsample of size $n_1$. Specifically, for $ n_0 $ as before, compute an initial estimate
  \begin{subequations}
      		\begin{align}
			\beta_0 &= \arg \min_{\beta}\bar{g}_{n_0} (\beta)'  \bar{g}_{n_0} (\beta)\\
			\Phi_0 &= \frac{\partial}{\partial \beta} \bar{g}_{n_0} (\beta_0) \\
            W_0 &= \left( \frac{1}{n_0} \sum_{j=1}^{n_0} g_{j}(\beta_0) g_{j}(\beta_0)' +\eta_0 I \right)^{-1} .
		\end{align}

  \end{subequations}
  Let $ G_i (\beta) = \frac{\partial}{\partial \beta} g_i (\beta) $.
	We  sequentially update from $i = 1$ until $i = n$:
	\begin{subequations}
	 	\begin{align}
	 	\beta_i &= \beta_{i-1} - \gamma_i ( \Phi_{i-1}'W_{i-1}\Phi_{i-1})^{\dagger}
      \Phi_{i-1}' W_{i-1} g_i(\beta_{i-1}),  \\ 	 		
	 	\Phi_i &= \Phi_{i-1} - \frac{1}{ n_0 + i}(\Phi_{i-1} - G_i (\bar\beta_{i-1})), \\
				m_i &= n_0+i-1+g_i(\bar\beta_{i-1})'W_{i-1} g_i(\bar\beta_{i-1}),  \\
		W_{i} &= \frac{n_0+i}{n_0+i-1} W_{i-1}  \left[ I- m_i^{-1}g_i(\bar\beta_{i-1})g_i(\bar\beta_{i-1})'W_{i-1} \right],  \\
	 	\bar \beta_i &=  \bar{\beta}_{i-1} - \frac{1}{i} (\bar{\beta}_{i-1}-\beta_i ).
	 \end{align}
	\end{subequations}
	We leave it to future research to derive the asymptotic properties of the above nonlinear SGMM. The online Sargan-Hansen test can be computed the same way as that in subsection \ref{Jtest}.

Second, the instruments considered in the paper are assumed to be valid and strong. Allowing many weak IVs in general nonlinear GMM models is an important question that has been fruitfully studied in the past two decades. 
We expect that overidentified Z-type stochastic approximation is appealing with many weak IVs because it can be helpful to deal  with the problem of many local minima. For instance, being a fast algorithm, overidentified Z-type stochastic approximation gives us an avenue of attempting many different starting values. However, we expect that the theoretical studies could be technically challenging, so we leave this extension for future research.  

\appendix

\section{Appendix}\label{sec:appendix}

\subsection{Proofs}

Throughout the proofs, with no loss of generality, assume $\beta_* = 0$ so that $\mathbb E[H_i] = 0$, under iid assumption $\mathbb E[H_i\mid\mathcal F_{i-1}] = 0$ and $\mathbb E[g_i(\beta_{i-1})\mid\mathcal F_{i-1}] = G\beta_{i-1}$.

A generic positive constant will be denoted by $K > 0$, whose value may differ in each occurrence. For future reference, for each $R > 0$, we consider the stopping times defined as follows.
\begin{equation}
\label{eq:tau_R:stopping-time}
    \begin{array}{rl}
    \tau_R &= \inf \{ i \ge 0: \|\beta_i\| \ge R\},\\[6pt]
    \sigma_R &= \inf \{ i \ge 0: \max\{\|W_i\|, \|(\Phi_{i-1}' W_{i-1} \Phi_{i-1})^{\dagger}\| \} \ge R\}, \\[6pt]
    \rho_R &= \inf \{ i \ge 0: \|\Phi_i - G\| / \eta_{i+1} \ge R \},\\[6pt]
    \iota_R &= \inf \{ i \ge 0: \|(\Phi_i' \Phi_i)^{\dagger}\| \ge R \},
    \end{array}        
\end{equation}
where $\eta_i = i^{-1/2} \log_+ \log_+ (i)^{1/2}$ for $\log_+(x) := \max\{ \log x, 1 \}$. We regard $\inf \varnothing = \infty$.

\begin{proof}[\textbf{Proof of Lemma~\ref{lem:strong consistency of beta_n}}]
Let us denote $B_i := \|\beta_i\|^ 2 $. Expanding the square $\|\beta_i\|^2 = \|\beta_{i-1} - \gamma_i (\Phi_{i-1}' W_{i-1} \Phi_{i-1})^{\dagger} \Phi_{i-1}' W_{i-1}g_i(\beta_{i-1})\|^2$, we have 
\begin{align}
\label{lem:basic:expansion:2sls:sgmm}
    B_i  &= B_{i-1} -  2 \gamma_i \beta_{i-1}' (\Phi_{i-1}' W_{i-1} \Phi_{i-1})^{\dagger} \Phi_{i-1}' W_{i-1}g_i(\beta_{i-1}) \\
    & \qquad + \gamma_i^ 2 g_i(\beta_{i-1})' W_{i-1}\Phi_{i-1} (\Phi_{i-1}' W_{i-1} \Phi_{i-1})^{\dagger}{}^{2} \Phi_{i-1}'W_{i-1} g_i(\beta_{i-1}) \nonumber.
\end{align}
Taking a conditional expectation $\mathbb E[\cdot\mid\mathcal F_{i-1}]$ on both sides, the second term on the right-hand side of \eqref{lem:basic:expansion:2sls:sgmm} yields $-2\gamma_i \beta_{i-1}' (\Phi_{i-1}' W_{i-1} \Phi_{i-1})^{\dagger} \Phi_{i-1}' W_{i-1}  G \beta_{i-1}$. 
For the third term, since we have 
\begin{align*}
    \opn{(\Phi_{i-1}' W_{i-1} \Phi_{i-1})^{\dagger}\Phi_{i-1}' W_{i-1}^{1/2}}^ 2 &=\opn{W_{i-1}^{1/2} \Phi_{i-1} (\Phi_{i-1}' W_{i-1} \Phi_{i-1})^{\dagger}{}^{2}\Phi_{i-1}' W_{i-1}^{1/2}} \\
    &= \opn{(\Phi_{i-1}' W_{i-1} \Phi_{i-1})^{\dagger}},
\end{align*}
it follows that 
\begin{align*}
    &\mathbb E [g_i(\beta_{i-1})' W_{i-1}\Phi_{i-1}(\Phi_{i-1}' W_{i-1} \Phi_{i-1})^{\dagger}{}^{2}\Phi_{i-1}' W_{i-1} g_i(\beta_{i-1})\mid \mathcal F_{i-1}] \\
    &\le  2\mathbb E[\beta_{i-1}'G_i' W_{i-1} \Phi_{i-1}(\Phi_{i-1}' W_{i-1} \Phi_{i-1})^{\dagger}{}^{2}\Phi_{i-1}' W_{i-1} G_i \beta_{i-1} \\
    & \qquad + H_i' W_{i-1} \Phi_{i-1} (\Phi_{i-1}' W_{i-1} \Phi_{i-1})^{\dagger}{}^{2}\Phi_{i-1}' W_{i-1}H_i \mid \mathcal F_{i-1}] \\
    &\lesssim \opn{W_{i-1}} \opn{(\Phi_{i-1}' W_{i-1} \Phi_{i-1})^{\dagger}}\mathbb E[\|\beta_{i-1}\|^ 2 \opn{G_i}^ 2 + \|H_i\|^2\mid \mathcal F_{i-1}]\\
    &\lesssim \opn{W_{i-1}} \opn{(\Phi_{i-1}' W_{i-1} \Phi_{i-1})^{\dagger}} (B_{i-1}+1)
\end{align*}
where the second-to-last inequality uses Assumptions \ref{A:sample} and \ref{A:moment:consistency}.
As a result, 
\begin{align}
\label{eq:middle-equation}
\mathbb E[B_i \mid \mathcal F_{i-1}] &\le B_{i-1} - 2\gamma_i \beta_{i-1}' (\Phi_{i-1}' W_{i-1} \Phi_{i-1})^{\dagger}\Phi_{i-1}' W_{i-1} G \beta_{i-1} \\
& \quad + K \gamma_i^ 2 \|{W_{i-1}}\|\|(\Phi_{i-1}' W_{i-1} \Phi_{i-1})^{\dagger}\| (B_{i-1} + 1) \nonumber.
\end{align}
for some $K>0$. Note that
\begin{align}
&-\beta_{i-1}' (\Phi_{i-1}' W_{i-1} \Phi_{i-1})^{\dagger}\Phi_{i-1}' W_{i-1} G \beta_{i-1} \label{neg} \\ 
= &-\beta_{i-1}' (\Phi_{i-1}' W_{i-1} \Phi_{i-1})^{\dagger}\Phi_{i-1}' W_{i-1} \Phi_{i-1} \beta_{i-1}  + \beta_{i-1}'(\Phi_{i-1}' W_{i-1} \Phi_{i-1})^{\dagger} \Phi_{i-1}' W_{i-1}(\Phi_{i-1} - G) \beta_{i-1}\nonumber\\
    \le & -\beta_{i-1}' (\Phi_{i-1}' W_{i-1} \Phi_{i-1})^{\dagger}\Phi_{i-1}' W_{i-1} \Phi_{i-1} \beta_{i-1} + \|\beta_{i-1}\|^ 2 \|(\Phi_{i-1}' W_{i-1} \Phi_{i-1})^{\dagger}\|^ {1/2} \| W_{i-1}\|^ {1/2} \|\Phi_{i-1} - G\|.\nonumber
\end{align}
Note that there exists $\delta > 0$ such that if $\|\Phi_{i-1} - G\| < \delta$ holds true, then the inverse matrix $(\Phi_{i-1}' W_{i-1} \Phi_{i-1})^{\dagger} = (\Phi_{i-1}' W_{i-1} \Phi_{i-1})^{-1}$ exists. As a result, on this event $\{ \|\Phi_{i-1} - G\| < \delta \}$, it holds
\begin{align*}
&-\beta_{i-1}' (\Phi_{i-1}' W_{i-1} \Phi_{i-1})^{\dagger}\Phi_{i-1}' W_{i-1} G \beta_{i-1} \\
    \le & -\|\beta_{i-1}\|^ 2 + \|\beta_{i-1}\|^ 2 \|(\Phi_{i-1}' W_{i-1} \Phi_{i-1})^{\dagger}\|^ {1/2} \| W_{i-1}\|^ {1/2} \|\Phi_{i-1} - G\|;
\end{align*}
whereas on the event $\{ \|\Phi_{i-1} - G\| \ge \delta \}$, \eqref{neg} implies trivially that
\begin{align*}
    & -\beta_{i-1}' (\Phi_{i-1}' W_{i-1} \Phi_{i-1})^{\dagger}\Phi_{i-1}' W_{i-1} G \beta_{i-1} \\
    \le &\ \|\beta_{i-1}\|^ 2 \|(\Phi_{i-1}' W_{i-1} \Phi_{i-1})^{\dagger}\|^{1/2}\|W_{i-1}\|^{1/2}\|\Phi_{i-1} - G\|.
\end{align*}
Putting this together with \eqref{eq:middle-equation}, it follows
\begin{align}
    \label{eq:B_i inequality in lem:strong consistency of beta_n:2sls:sgmm}
    &\mathbb E[B_i\mid \mathcal F_{i-1}]  \\
    \le &\ B_{i-1} \big( 1+ 2\gamma_i \|\Phi_{i-1}-G\| \|W_{i-1}\|^ {1/2} \|(\Phi_{i-1}' W_{i-1} \Phi_{i-1})^{\dagger}\|^ {1/2} + K \gamma_i^ 2 \|W_{i-1}\| \|(\Phi_{i-1}' W_{i-1} \Phi_{i-1})^{\dagger}\| \big)\nonumber\\
    & + K \gamma_i^ 2 \|W_{i-1}\| \|(\Phi_{i-1}' W_{i-1} \Phi_{i-1})^{\dagger}\|- 2\gamma_i B_{i-1} 1\{ \|\Phi_{i-1} - G\| < \delta \}  \nonumber
\end{align}
for some $K > 0$, where $1_A$ denotes an indicator function for a set $A$. 

By Assumption \ref{A:moment:consistency} and the law of iterated logarithm (LIL), $\limsup_{i\to\infty}{\|\Phi_{i-1} - G\|}/{\eta_i} \le c $ holds almost surely for some $c>0$, which implies that $1\{ \|\Phi_{i-1} - G\| < \delta \} = 1$ holds for all but finitely many $i$ almost surely. Also, using the fact that $\lim_{i\to\infty}W_{i-1} = W$ and $\lim_{i\to\infty}(\Phi_{i-1}' W_{i-1} \Phi_{i-1})^{\dagger} = (G'W G)^{-1}$ by the strong law of large numbers (SLLN), we have
\begin{align*}
    \sum_{i=1}^\infty (2\gamma_i \|\Phi_{i-1}-G\| \|W_{i-1}\|^ {1/2} \|(\Phi_{i-1}' W_{i-1} \Phi_{i-1})^{\dagger}\|^ {1/2} + K \gamma_i^ 2 \|W_{i-1}\| \|(\Phi_{i-1}' W_{i-1} \Phi_{i-1})^{\dagger}\|) < \infty
\end{align*}
almost surely (a.s.)

By Lemma \ref{lem:Robbins Siegmund}, it follows that $\lim_{i\to\infty}B_i = B_\infty <\infty$ and $\sum_{i=1}^\infty \gamma_i B_{i-1} 1\{ \|\Phi_{i-1} - G\| < \delta \} <\infty$ exist a.s. This implies that $B_\infty = 0$ a.s., because otherwise $\sum_i \gamma_i <\infty$, which is in contradiction to Assumption that $\gamma_i=\gamma_0 i^{-a}$ with $a \in (1/2,1]$.
Therefore, $\lim_{n\to\infty}\|\beta_n\| = 0$ a.s., and we conclude that $\beta_n \to \beta_*$ as $n\to\infty$ a.s..
\end{proof}

\begin{proof}[\textbf{Proof of Theorem~\ref{thm:1}}]
\phantom{}\\
\noindent\textbf{\textit{Part 1.}} 
Local $L^2$-convergence rate.

Consider the stopping time $T_R = \min\{ \sigma_R, \rho_R \}$ defined as per \eqref{eq:tau_R:stopping-time}. In Part~1, we aim to establish the convergence rate of $\mathbb E[\|\beta_i\|^21\{T_R\geq i\}]$ to 0. The proof stems from (\ref{eq:B_i inequality in lem:strong consistency of beta_n:2sls:sgmm}) with the fact that the event $T_R\geq i$ implies $\|W_{i-1}\| \le R$, $\|(\Phi_{i-1}' W_{i-1} \Phi_{i-1})^{\dagger}\| \le R$, and $\|\Phi_{i-1} - G\| \le R \eta_i$.

Note that $\sup_{i\ge 0} \|W_i\| < \infty$ and $\sup_{i\ge 0} \|(\Phi_{i}' W_{i} \Phi_{i})^{\dagger}\| < \infty$ holds a.s. by the SLLN, and also that $\sup_{i \ge 0} \|\Phi_{i} - G \|/\eta_{i+1} < \infty$ a.s. by the LIL. Thus, it follows $\mathbb P(T_R = \infty\text{ for some }R>0) = 1$. From \eqref{eq:B_i inequality in lem:strong consistency of beta_n:2sls:sgmm} in the proof of Lemma \ref{lem:strong consistency of beta_n}, we can see that on $\{ T_R \ge i \}$, it holds
$$
\|\Phi_{i-1}-G\| \|W_{i-1}\|^ {1/2} \|(\Phi_{i-1}' W_{i-1} \Phi_{i-1})^{\dagger}\|^ {1/2}
\leq R^2 \eta_i.
$$
Thus, for all sufficiently large $i$ such that $\eta_i \le \frac{1}{2R^2}$, it holds
\begin{equation}\label{eq4:2sls:sgmm}
    \mathbb E[B_i 1\{T_R \geq  i\}\mid \mathcal F_{i-1}] \le B_{i-1}1\{T_R \geq i-1\} (1-\alpha \gamma_i + K\gamma_i^ 2) + K\gamma_i^ 2    
\end{equation}
for some $\alpha > 0$. By integrating both sides of this inequality, we obtain
\begin{equation*}
    \mathbb E[B_i 1\{T_R > i\}] \le \mathbb E[B_{i-1}1\{T_R > i- 1\}] (1 - \alpha \gamma_i + K\gamma_i^ 2) + K\gamma_i^ 2.
\end{equation*}
By Lemma \ref{lem:recursive bound}, it follows that $\mathbb E[B_i 1\{T_R > i\}] = O(\gamma_i)$ for any choice of $R>0$.

\noindent\textbf{\textit{Part 2.}} Coupling with the linearized process $\beta_i^1$.

Define an auxiliary process $\beta_i^1$ as follows.
\begin{equation}
    \label{eq:definition of beta_n^1:2sls:sgmm}
    \beta_i^ 1 := (1 - \gamma_i)\beta_{i-1}^ 1 - \gamma_i \xi_i
\end{equation}
where $\beta_0^ 1 := \beta_0$, $\tilde G_i := G_i - G$ and
\begin{equation}\label{xi-def}
\xi_i:=(\Phi_{i-1}' W_{i-1} \Phi_{i-1})^{\dagger}\Phi_{i-1}' W_{i-1}(\tilde G_i \beta_{i-1} + H_i).
\end{equation}
Also define $\bar \beta_i^ 1 = \frac 1 i \sum_{j=1}^i \beta_j^ 1$ analogously to $\bar\beta_i$. 

Define $\delta_i =  \beta_i - \beta_i^ 1$ as the difference between $\beta_i$ and its approximation $\beta_i^1$. By Lemma \ref{lem:uniform approximation by beta_n^1}, it follows that $\sqrt n (\bar \beta_1 - \bar\beta_n^ 1) = \frac 1 {\sqrt n} \sum_{i=1}^ n \delta_i = o_P(1)$, wherein we used the local convergence rate established in Part~1. Hence, we have $\sqrt n \bar \beta_n = \sqrt n \bar \beta_n^ 1 + o_P(1)$. 

\noindent \textbf{\textit{Part 3.}}
We now assert that the sequence $\beta_i^1$ is the SGD sequence of 
Polyak and Juditsky's~\citeyearpar[PJ hereafter]{polyak1992acceleration} Theorem~1-(a) by verifying its Assumptions~2.1 to~2.5-(a).
First of all, the ``true parameter'' for this sequence is zero because we have assumed $\beta_*=0$. Assumption~2.1 is straightforward to check because $A = I$ in our case. Secondly, $\xi_i$ defined in \eqref{xi-def} is a martingale difference sequence (mds) because $\mathbb E[H_i\mid \mathcal F_{i-1}]= H=-G\beta_*=0$ and 
$$
\mathbb E[\xi_i \mid  \mathcal F_{i-1}]=\Phi_{i-1}' W_{i-1} (\mathbb E[\tilde G_i\mid \mathcal F_{i-1}] \beta_{i-1} + H)=0.
$$
This deals with Assumption~2.2. In addition, Assumptions 2.3-2.5(a) of PJ are verified in Lemma \ref{lem:Assumptions 2.3, 2.4, 2.5-(a) from PJ (1992)}.
It then follows from Theorem~1-(a) in \cite{polyak1992acceleration}  that $\sqrt n \bar \beta_n^ 1 \overset{d}{\to} N(0, (G' W G)^{-1} G'W\Omega W G (G'W G)^{-1})$. 
We conclude $\sqrt n\bar \beta_n \overset{d}{\to} N(0, (G'WG)^{-1} G'W\Omega W G (G'WG)^{-1})$. 

\end{proof}

\begin{proof}[\textbf{Proof of Theorem~\ref{thm:2}}] We maintain the assumption $\beta_* = 0$. 
Let us define
\begin{align}
\label{eq:define-nu}
    \bar \nu(r) &= \frac 1{\sqrt n} \sum_{i=0}^ {\lfloor nr\rfloor} \beta_i, \\
    \bar \nu^ 1(r) &= \frac 1 {\sqrt n} \sum_{i=0}^ {\lfloor nr\rfloor} \beta^ 1_i,\nonumber
\end{align}
for $r \in [0,1]$ and $\beta^ 1$ as defined in the proof of Theorem \ref{thm:1}. By Lemma \ref{lem:uniform approximation by beta_n^1}, we have that $\sup_{r \in [0,1]}\|\bar \nu(r) - \bar \nu^ 1(r)\|  = \frac 1 {\sqrt n} \sup_{1\le m\le n} \left\|\sum_{i=1}^ m (\beta_i -\beta_i^ 1)\right\|= o_P(1)$.  This allows us to prove Theorem \ref{thm:2} with $\bar \nu^ 1(r)$ in lieu of $\bar \nu(r)$. 

Now, we consider a decomposition $\bar \nu^ 1(r) = I_1(r) - I_2(r) - I_3(r)$, defined as
\begin{align*}
    I_1(r) &= \frac 1{\sqrt n \gamma_0} \alpha_0^{\lfloor nr\rfloor} \beta_0   ,\\
    I_2(r) &= \frac 1 {\sqrt n} \sum_{i=1}^{\lfloor nr \rfloor} \xi_i  ,\\
    I_3(r) &= \frac 1 {\sqrt n} \sum_{i=1}^{\lfloor nr \rfloor} w^ {\lfloor nr \rfloor}_i \xi_i ,
\end{align*}
with $\xi_i$ defined in \eqref{xi-def}, $\alpha_j^ n := \gamma_j \sum_{i=j}^ n \prod_{k=j+1}^ i (1 - \gamma_k)$ and $w_j^ m := \alpha_j^ m - 1$. Observe that $|\alpha_j^ n| \lesssim 1$ holds uniformly for all $j\le n$ by Lemma~1-(ii) in \cite{polyak1992acceleration}, and hence $|w_j^ m |\lesssim 1$ also holds uniformly. Moreover, we have $\sum_{j=1}^m |w_j^ m| = O(m^ {a})$ \citep{Zhu:Dong}.

Since $\alpha_0^ m$ is uniformly bounded in $m$, it is easy to see that $\sup_r \|I_1(r)\| = o_P(1)$. For $I_2(r)$,  
a standard FCLT for mds applies and yields $\{I_2(r)\}_{r\in[0,1]}\rightsquigarrow \Avar(\bar \beta)^{1/2}\{W_{d_\beta}(r)\}_{r\in[0,1]}$. Note that sufficient conditions for FCLT are provided by Lemma \ref{lem:Assumptions 2.3, 2.4, 2.5-(a) from PJ (1992)}.

For the third term, we split $I_3(r)$ into two components:
\begin{align*}
    I_3(r) &= \frac 1 {\sqrt n} \sum_{i=1}^{\lfloor nr \rfloor} w^ {\lfloor nr \rfloor}_i (\Phi_{i-1}' W_{i-1} \Phi_{i-1})^{\dagger}\Phi_{i-1}'W_{i-1}\tilde G_i \beta_{i-1} \\
    &\quad +  \frac 1 {\sqrt n} \sum_{i=1}^{\lfloor nr \rfloor} w^ {\lfloor nr \rfloor}_i (\Phi_{i-1}' W_{i-1} \Phi_{i-1})^{\dagger}\Phi_{i-1}'W_{i-1}H_i\\
    & =: I_{31}(r) + I_{32}(r).    
\end{align*}
Under Assumption \ref{A:moment:FCLT}, the treatment of $\sup_r \|I_{32}(r)\| = o_P(1)$ can be approached in a similar manner to that in \cite{lee2021fast} for their linear least squares regression, which will be omitted here. However, the term $I_{31}(r)$ cannot be handled in the same manner due to the absence of ex-ante moment conditions for $\beta_{i-1}$. 
To reach a proper bound for $\beta_{i-1}$, we consider $T_R = \min\{ \tau_R, \sigma_R \}$. Then on the event $\{ T_R\geq i \}$, we have $\|\beta_{i-1}\|\leq {R}$, $\|(\Phi_{i-1}' W_{i-1} \Phi_{i-1})^{\dagger}\|\le R$ and $\|W_{i-1}\| \le R$.

Let us denote $S_m := \sum_{i=1}^m w_i^ m (\Phi_{i-1}' W_{i-1} \Phi_{i-1})^{\dagger}\Phi_{i-1}' W_{i-1} \tilde G_i \beta_{i-1}$ and $S_{m\wedge T_R} := S_{\min\{m, T_R\}}$ for $m \ge 1$. We examine the first component of $I_3(r)$, which satisfies
$$
\sup_{r \in [0,1]} \|I_{31}(r)\| = n^{-1/2} \sup_{1\le m\le n} \|S_m\|.
$$
Let $p$ be an integer such that $p > (1-a)^{-1}$ (see Assumption \ref{A:moment:FCLT}). Note that on the event $\{T_R > n\}$, it holds
$$
\sup_{1\le m\le n} \|S_m\|^{2p} \le \sum_{m=1}^n \|S_m\|^{2p} = \sum_{m=1}^ n \|S_{m\wedge T_R}\|^{2p}
$$
From this, we have
$$
\mathbb E[\sup_r \|I_{31}(r)\|^{2p}1\{T_R > n\}] \le n^{-p} \sum_{m=1}^n \mathbb E[\|S_{m\wedge T_R}\|^{2p}].
$$

Write $A_i:=  (\Phi_{i-1}' W_{i-1} \Phi_{i-1})^{\dagger}\Phi_{i-1}' W_{i-1}\tilde G_i \beta_{i-1}1\{T_R \ge i\}$. Note that $\{T_R\geq i\}$ is equivalent to $\{\max_{j\leq i-1}\{ \|\beta_j\|, \|(\Phi_{j}' W_{j} \Phi_{j})^{-1}\|, \|W_j\| \}<R\}$, which belongs to $\mathcal F_{i-1}$. Hence, $w_i^m A_i$ constitutes a mds, and we can express $S_{m\wedge T_R}$ as follows:
$$
S_{m\wedge T_R} = \sum_{i=1}^m w_i^ m (\Phi_{i-1}' W_{i-1} \Phi_{i-1})^{\dagger}\Phi_{i-1}' W_{i-1} \tilde G_i \beta_{i-1}1\{T_R \ge i\}
=\sum_{i=1}^m w_i^ m  A_i.
$$
Also note that $\|A_i\|\leq R^{2}\| \tilde G_i \|$. 
Applying Burkholder's inequality \citep[e.g.,][]{Hall-Heyde} to $\|\sum_{i=1}^m w_i^ m  A_i\|^{2p}$, we get
\begin{align*}
    & \mathbb E[\sup_r \|I_{31}(r)\|^{2p}1\{T_R > n\}] \\
    \le &\  n^{-p} \sum_{m=1}^n \mathbb E[\|S_{m\wedge T_R}\|^{2p}] 
=\  n^{-p} \sum_{m=1}^n \mathbb E[\|
    \sum_{i=1}^mw_i^mA_i\|^{2p}] 
\\
\le &\  n^{-p} \sum_{m=1}^n \mathbb E\left(
    \sum_{i=1}^m\|w_i^mA_i\|^{2} \right)^p
\\
\le &\  n^{-p} \sum_{m=1}^n 
    \sum_{i_1,...,i_p}^m|w_{i_1}^m|^2 ...|w_{i_p}^m|^2\mathbb E\|A_{i_1}\|^{2}... \|A_{i_p}\|^{2}  
\\
\le &\  n^{-p} \sum_{m=1}^n 
    \sum_{i_1,...,i_p}^m|w_{i_1}^m|^2 ...|w_{i_p}^m|^2\max_{j\leq m}\mathbb E\|A_{j}\|^{2p}\cdot p
\\
\lesssim &\ n^{-p} \sum_{m=1}^n \left(\sum_{i=1}^m |w^ m_{i}| \right)^ p
\lesssim \ n^{-p} \sum_{m=1}^n m^{ap} \lesssim   n^{-(1-a)p + 1}\to 0
\end{align*}
where Young's inequality $\prod_{r=1}^p \|a_{r}\| \lesssim \sum_{r=1}^p \|a_r\|^{p}$ is invoked for deriving the second-to-last line.
Hence, $\sup_r \|I_{31}(r)\| 1\{T_R > n\} = o_P(1)$ as $n\to \infty$. By the same argument as used in the proof of Lemma \ref{lem:uniform approximation by beta_n^1}, we have $\sup_r \|I_{31}(r)\|\to_p 0$. We conclude $\{\bar \nu(r)\}_{r\in [0,1]} \rightsquigarrow \Avar(\bar\beta)^{1/2}\{W_{d_\beta}(r)\}_{r\in[0,1]}.$
\end{proof}

\begin{proof}[\textbf{Proof of Theorem~\ref{thm:3}}]

The proof of Theorem \ref{thm:3} differs from that of Theorem \ref{thm:1} under the fully online setting as there is a need to deal with the triangular array $(W_{i,n})_{i=0}^{n-1}$ of weighting matrices, as defined in \eqref{algo:beta:W:eff}. It is worth noting that this introduces a dependence of $(\beta_{i,n})_{i=1}^n$ on $n$, emerging due to the varying size $n_1 = n_1(n)$ of $\mathbb{S}_1$ as the sample size $n$ changes. Since $n_1$ may vary with $n$, it is possible to have $W_{i,n} \ne W_{i,m}$ for $n \ne m$, posing a challenge when analyzing the behavior of $(W_{i,n})_{i=0}^{n-1}$ using previous approaches. As a result, an alternative technique is employed to handle the triangular array structure in this proof.

\noindent \textbf{\textit{Part 1}}. $L^2$-convergence rate and consistency of SGD from the efficient algorithm \eqref{algo:beta:eff}--\eqref{algo:betabar:eff}.

We will first derive the uniform $L^2$-convergence rate for the class of all \textit{admissible} weighting schemes $(W_{i})_{i=0}^{n-1}$. To this end, we allow $W_{i}$ to be a possibly random positive definite matrix adapted to the filtration $\mathcal{F}_{i}$. For instance, in the updating rule in \eqref{algo:beta:W:2sls}, $W_{i}$ is treated as a unit point mass at 
$
\frac{i}{i-1}( W_{i-1} - \frac{1}{m_i}W_{i-1}z_i z_i' W_{i-1}),
$
which is $\mathcal{F}_{i}$-measurable.
Note that admissible weighting schemes allow for a broader range of possibilities than the proposed updating rule.

We denote $W_{[n]}$ as the sequence of weighting matrices $(W_{i})_{i=0}^{n-1}$ up to $n-1$ and say that $W_{[n]} \in \mathcal{A}$ if $W_{[n]}$ is an admissible weighting scheme. To differentiate the proposed updating scheme from a generic one, we denote $W_{i,n}$ as the weighting matrix following the rule \eqref{algo:beta:W:eff} with $n_1(n)$ as the change-point. For $C \ge 1$, define an event $\{ W_{[k]} \in \mathcal{E}(C) \}\in \mathcal{F}_{k-1}$ as 
\begin{equation*}
    \{ W_{[k]} \in \mathcal{E}(C) \}= \{1/C \le \lambda_{\min}(W_j) \le \lambda_{\max}(W_j) \le C,\ j =1,\ldots, k-1\}.
\end{equation*}

For given $R > 0$ and $C \ge 1$, we consider the sequence
\begin{equation*}
    d_i := \sup_{W_{[i]} \in \mathcal{A}}  \mathbb{E} [\|\beta_i\|^ 2 1\{ T_R \ge i, W_{[i]} \in \mathcal{E}(C) \}].
\end{equation*}
where $T_R := \min \{ \rho_R, \iota_R \}$ and the supremum is taken over all admissible weighting schemes. Note that $d_0 < \infty$ is well-defined because $\mathbb{E}[ \|\beta_0\|^ 2] < \infty$. By the inductive step below, we can also see that $d_i < \infty$ for all $i \ge 0$. We aim to establish $d_n = O(\gamma_n)$. Note that on an event $\{ T_R \ge i, W_{[i]} \in \mathcal{E}(C) \} \in \mathcal{F}_{i-1}$ where $W_{[i]} \in \mathcal{A}$, it follows from \eqref{eq:B_i inequality in lem:strong consistency of beta_n:2sls:sgmm} that 
\begin{align*}
    \mathbb{E}[\|\beta_i\|^ 2 \mid \mathcal{F}_{i-1}] &\le \|\beta_{i-1}\|^ 2 (1 -  \gamma_i/2 + K\gamma_i^2) + K \gamma_i^ 2
\end{align*}
for all sufficiently large $i$. Since $\{ T_R \ge i, W_{[i]} \in \mathcal{E}(C) \} \subseteq \{ T_R \ge i-1, W_{[i-1]} \in \mathcal{E}(C) \}$, we obtain 
\begin{align*}
    &\mathbb{E} [\|\beta_i\|^ 2 1\{ T_R \ge i, W_{[i]} \in \mathcal{E}(C) \}] \\
    \le &\ \mathbb{E} [\|\beta_{i-1}\|^ 2 1\{ T_R \ge i-1, W_{[i-1]} \in \mathcal{E}(C) \}] (1- \gamma_i/2 + K\gamma_i^2) + K \gamma_i^ 2,  
\end{align*}
whence it follows 
\begin{equation*}
    d_i \le d_{i-1} (1- \gamma_i/2 + K \gamma_i^2) + K \gamma_i^ 2
\end{equation*}
for all sufficiently large $i$ by taking supremum over $W_{[i]}$ on both sides. By Lemma \ref{lem:recursive bound}, it follows $d_n = O(\gamma_n)$ as $n\to\infty$. 

By specializing this to $W_{[n], n}$, which is trivially admissible, we obtain
\begin{equation*}
    \mathbb{E} [\|\beta_n\|^2 1\{ T_R \ge n, W_{[n],n} \in \mathcal{E}(C)\}] \le O(\gamma_n) \to 0
\end{equation*}
and 
\begin{equation*}
    \mathbb{E} \left[\left\|\frac 1 n \sum_{i=1}^n \beta_i 1\{ T_R \ge i, W_{[i],n} \in \mathcal{E}(C) \}\right\|^2\right] \to 0.    
\end{equation*}
Note that, thus far, the convergence is uniform in the choice of the change-point $n_1$ for any given $R$ and $C$. 

Next, our goal is to establish that $\beta_{n, n} = o_P(1)$ and $\bar\beta_{n,n} = o_P(1)$ when $n_1$ is chosen such that $n_1 \to \infty$ as $n \to \infty$. We drop the duplicate $n$ in the subscripts and denote them just by $\beta_{n}$ and $\bar\beta_{n}$ for brevity of notation. Accordingly, $W_{i,n}$ will also be denoted as $W_i$ occasionally for notational ease. 

Given that $n_1$ is chosen such that $n_1(n) \to \infty$, we verify that it holds
\begin{equation}
    \label{eq:W-matrix-bounded}
    \lim_{C \to \infty} \liminf_{n \to \infty} \mathbb{P} (W_{[n],n} \in \mathcal{E}(C)) = 1,
\end{equation}
i.e., $\mathbb{P} (W_{[n],n} \in \mathcal{E}(C))$ can be made arbitrarily close to 1 uniformly in $n$ for a sufficiently large $C$. Once this is verified, coupled with the fact that $\sup_{n \in \mathbb{N}} \mathbb{P}(T_R \le n) = \mathbb{P}(T_R < \infty) < \varepsilon$ for sufficiently large $R>0$, it implies that
\begin{align*}
    \mathbb{P}( \| \beta_n - \beta_n 1\{ T_R \ge n, W_{[n],n} \in \mathcal{E}(C)\} \| > \varepsilon ) < 2\varepsilon,
\end{align*}
and 
\begin{equation*}
    \mathbb{P}\left( \left\|\bar\beta_n - \frac 1 n \sum_{i=1}^n \beta_i 1\{ T_R \ge i, W_{[i],n} \in \mathcal{E}(C) \}\right\| > \varepsilon \right) < 2 \varepsilon,
\end{equation*}
which establishes $\beta_n = \beta_n 1\{ T_R > n, W_{[n],n} \in \mathcal{E}(C)\} + o_P(1)= o_P(1)$ and $\bar \beta_n = o_P(1)$.

We start with a convenient observation. Define $W^*_i := W_{i,n}$ for $n$ such that $i \le n_1(n)$. Note that there is no ambiguity in this definition since $W_{i,n} = W_{i,m}$ whenever $n_1(n)\ge i$ and $n_1(m) \ge i$ hold. This is well-defined as we assumed $n_1 \to \infty$ and hence $W_i^* = \lim_n W_{i, n}$ holds. Since $W^*_{r}  = (\frac{n_0}{n_0 + r}W_0^{-1} + \frac{1}{n_0 + r} \sum_{j=1}^{r} z_j z_j')^{-1}$, $r\ge 1$ converges a.s. to a positive definite matrix $Q^{-1}$ as $r \to \infty$ by the SLLN, we can see that with probability 1, $(\lambda_{\min}(W^*_i))_{i=1}^\infty$ is bounded away from 0 and $(\lambda_{\max}(W^*_i))_{i=1}^\infty$ is bounded from above. This implies that
\begin{align*}
   &\lim_{C\to\infty}\liminf_{n\to\infty} \mathbb{P}(C^{-1} \le \lambda_{\min}(W_{j,n}) \le \lambda_{\max}(W_{j,n}) \le C, j = 0,\ldots, n_1) \\
   \ge &\ \lim_{C\to\infty} \mathbb{P}(C^{-1} \le \lambda_{\min}(W_{j}^*) \le \lambda_{\max}(W_{j}^*) \le C, j = 0,1,\ldots)=1.
\end{align*}
This observation allows us to establish only
\begin{equation}
    \label{eq:W-matrix-bounded'}
\lim_{C\to \infty}\liminf_{n\to\infty}\mathbb{P}(C^{-1} \le \lambda_{\min}(W_{j,n}) \le \lambda_{\max}(W_{j,n}) \le C, n_1 < j \le n) =1
\end{equation}
to prove \eqref{eq:W-matrix-bounded}.

Furthermore, it is useful to note that, for $\beta^*_i := \beta_{i,n}$ and $\bar\beta^*_i := \bar\beta_{i,n}$ defined in a similar manner to that of $W_{i}^*$, we have $\beta^*_{r} \to 0$ and $\bar\beta^*_{r} \to 0$ as $r$ tends to infinity by Lemma \ref{lem:strong consistency of beta_n}. In particular, this implies the strong consistency of $\bar\beta_{n_1} = \bar\beta_{n_1}^*$ as $n\to\infty$, which will be utilized in the subsequent analysis.

To verify \eqref{eq:W-matrix-bounded'}, note that 
\begin{align*}
    & \mathbb{P}(C^{-1} \le \lambda_{\min}(W_{j,n}) \le \lambda_{\max}(W_{j,n}) \le C, j = n_1 + 1,\ldots, n )\\ 
    = &\ \mathbb{P}\left(\left\{ C^{-1} \le \min_{n_1 < j\le n}\lambda_{\min}(W_{j,n}^{-1})  \right\} {\textstyle\bigcap}\left\{  \max_{n_1 < j\le n}\lambda_{\max}(W_{j,n}^{-1}) \le C \right\}\right)\\
    \ge & \ \mathbb{P}\left(C^{-1} \le \min_{n_1 < j\le n}\lambda_{\min}(W_{j,n}^{-1}) \right) + \mathbb{P}\left(\max_{n_1 < j\le n}\lambda_{\max}(W_{j,n}^{-1}) \le C\right)-1 \\
    =: &\ P_{1}(C) + P_{2}(C)-1.
\end{align*}
We will show that $\lim_{C\to\infty}\liminf_{n\to\infty} P_{j}(C) = 1$ for each $j = 1,2$, which then establishes \eqref{eq:W-matrix-bounded'}.

We first address $P_1(C)$. For brevity of notation, let us denote $\mathbb{E}_{m:n}[X_j] = \frac{1}{n-m} \sum_{j=m+1}^n X_j$ as the empirical average of $X_j$ from $j = m+1$ to $n$. 
We observe that 
\begin{align*}
    & W_{i,n}^{-1} = \frac{n_0 + n_1} {n_0 + n_1 + k} W_{n_1}^{-1}  + \frac{k} {n_0 + n_1 + k}\mathbb{E}_{n_1:n_1 + k}[g_j(\bar\beta_{n_1}) g_j(\bar \beta_{n_1})']
\end{align*}
where $k := i- n_1 \in \{ 1,\ldots, n_2 \}$ and $n_2 := n - n_1(n)$. Thus, each $W_{i,n}^{-1}$ is given by a weighted average of $W_{n_1,n}^{-1} = W^*_{n_1}{}^{-1}$ and $\mathbb{E}_{n_1:n_1+k}[g_j(\bar\beta_{n_1}) g_j(\bar \beta_{n_1})']$ for some $1\le k = i - n_1 \le n_2$.

For arbitrary integer $m \ge d_g$, note that 
\begin{align*}
    &  P_1(C) \\
    = &\ \mathbb{P}(C^{-1} \le \min_{n_1 < j\le n}\lambda_{\min}(W_{j,n}^{-1}))\\
    \ge &\ \mathbb{P} \left( \left\{ \frac{n_0 + n_1 + m}{n_0 + n_1} C^{-1}\le \lambda_{\min}(W_{n_1,n}^{-1})\right\} {\textstyle\bigcap}\left\{ C^{-1} \le \min_{m\le k\le n_2}\lambda_{\min}(\mathbb{E}_{n_1:n_1+k}[g_j(\bar\beta_{n_1}) g_j(\bar\beta_{n_1})']\right\} \right) \\
    \ge &\ \mathbb{P} \left(\frac{n_0 + n_1 + m}{n_0 + n_1} C^{-1} \le  \lambda_{\min}(W^*_{n_1}{}^{-1})\right)+ \mathbb{P}\left( C^{-1} \le \min_{m\le k\le n_2}\lambda_{\min}(\mathbb{E}_{n_1:n_1+k}[g_j(\bar\beta_{n_1}) g_j(\bar\beta_{n_1})']\right) - 1 \\
    =: & I_{1}(m, C) + I_2(m, C) - 1.
\end{align*}
Here, we consider $k \ge m \ge d_g$ to prevent rank deficiency of the matrix $\mathbb{E}_{n_1:n_1+k}[g_j(\bar\beta_{n_1}) g_j(\bar\beta_{n_1})']$. We will verify that $\lim_{m\to\infty}\lim_{C \to \infty}\liminf_{n\to\infty}I_j(m, C) = 1$ for $j=1,2$, which then implies that $\lim_{C\to\infty}\liminf_{n} P_1(C) = 1$. 

The first term $I_1(m, C)$ is straightforward to deal with, because, by the previous observation that $W^*_{n_1} \to Q^{-1}$ a.s., it holds
\begin{align*}
\liminf_{n_1\to\infty}\mathbb{P} \left(\frac{n_0 + n_1 + m}{n_0 + n_1} C^{-1} \le  \lambda_{\min}(W^*_{n_1}{}^{-1})\right) &\ge \mathbb{P} \left(2 C^{-1} \le  \lambda_{\min}(Q)\right),
\end{align*}
implying $\lim_{m\to\infty}\lim_{C \to \infty}\liminf_{n\to\infty}I_1(m, C) \ge \mathbb{P} \left(0 <  \lambda_{\min}(Q)\right) = 1$. 

For the second term $I_2(m, C)$, we take advantage of the following fact; conditional on $\bar \beta_{n_1} = \beta$, the distribution of $(g_j(\bar\beta_{n_1}))_{j=n_1 + 1}^n$ is the same as the (unconditional) distribution of $(g_j(\beta)))_{j=1}^{n_2}$. Let $\mu_{n_1}(\cdot)$ denote the probability measure on $\mathbb{R}^{d_\beta}$ corresponding to the distribution of $\bar\beta_{n_1}$. Then, by the law of iterated expectations,
\begin{align*}
    I_2(m, C) &= \mathbb{E}_{\bar\beta_{n_1}}\left[ \mathbb{P}\left( C^{-1} \le \min_{m\le k\le n_2}\lambda_{\min}(\mathbb{E}_{n_1:n_1+k}[g_j(\bar\beta_{n_1}) g_j(\bar\beta_{n_1})'] \mid \bar\beta_{n_1}\right)\right]\\
    &\ge \int_{\mathbb{R}^{d_\beta}} \mathbb{P}\left(C^{-1} \le \inf_{m\le k}\lambda_{\min}(\mathbb{E}_{0:k}[g_j(\beta) g_j(\beta)']\right)d\mu_{n_1}(\beta).
\end{align*}
Since $\mu_{n_1}(\mathcal{K}) \to 1$ by the consistency of $\bar\beta_{n_1}$, it follows 
\begin{align*}
    \liminf_{n\to\infty} I_2(m, C) & \ge \liminf_{n\to\infty}  \int_{\mathcal{K}} \mathbb{P}\left(C^{-1} \le \inf_{m\le k}\lambda_{\min}(\mathbb{E}_{0:k}[g_j(\beta) g_j(\beta)']\right)d\mu_{n_1}(\beta) \\
    &\ge \inf_{\beta \in \mathcal{K}}  \mathbb{P}\left(C^{-1} \le \inf_{m\le k}\lambda_{\min}(\mathbb{E}_{0:k}[g_j(\beta) g_j(\beta)']\right) \\
    &\ge  \mathbb{P}\left(C^{-1} \le \inf_{m\le k}\inf_{\beta\in\mathcal{K}}\lambda_{\min}(\mathbb{E}_{0:k}[g_j(\beta) g_j(\beta)']\right).
\end{align*}
By the monotonicity in $m$ and $C$ as $m, C \to \infty$, the last probability tends to
\begin{align*}
    &\lim_{m \to \infty, C\to\infty} \mathbb{P}\left(C^{-1} \le \inf_{m\le k}\inf_{\beta\in\mathcal{K}}\lambda_{\min}(\mathbb{E}_{0:k}[g_j(\beta) g_j(\beta)']\right) \\
    =&\ \mathbb{P}\left(0 < \liminf_{k \to \infty}\inf_{\beta\in\mathcal{K}}\lambda_{\min}(\mathbb{E}_{0:k}[g_j(\beta) g_j(\beta)']\right).
\end{align*}
By the ULLN applied to $\mathbb{E}_{0:k}[g_j(\beta)g_j(\beta)'] = \frac 1 k \sum_{j=1}^k g_j(\beta)g_j(\beta)'$ indexed by ${\mathcal{K}}$, we have 
\begin{equation*}
   \inf_{\beta \in \mathcal{K}}\lambda_{\min}(\mathbb{E}_{0:k}[g_j(\beta)g_j(\beta)']) \to \inf_{\beta \in \mathcal{K}}\lambda_{\min}(\mathbb{E}[g_j(\beta)g_j(\beta)']) \ge c  > 0\quad\text{a.s.}
\end{equation*}
as $k \to \infty$ by the assumption. It follows that 
\begin{equation*}
\lim_{m, C \to \infty} \liminf_{n\to\infty}I_2(m, C) = 1
\end{equation*}
completing the proof of $\lim_{C\to \infty}\liminf_n P_1(C) = 1$.

The treatment of $P_2(C)$ is similar to that of $P_1(C)$ upon noticing that 
\begin{align*}
    & 1 - P_2(C) \\    
    \ge &\ \mathbb{P} \left( \lambda_{\max}(W^*_{n_1}{}^{-1}) \le C\right)+ \mathbb{P}\left( \max_{1\le k\le n_2}\lambda_{\max}(\mathbb{E}_{n_1:n_1+k}[g_j(\bar\beta_{n_1}) g_j(\bar\beta_{n_1})'] \le C\right) - 1 \\
    =: & \ I_1(C) + I_2(C) - 1.
\end{align*}

One can establish $\lim_{C\to\infty}\liminf_{n\to\infty}I_j(C) = 1$ for $j = 1,2,$ by the same argument as before, which will be omitted for brevity. We conclude that \eqref{eq:W-matrix-bounded'} is true, and so is \eqref{eq:W-matrix-bounded}.

\noindent\textit{\textbf{Part 2}}. Coupling with the linearized process $\beta_n^1$.

Define $\beta_i^1$ $(0\le i\le n)$ as $\beta_i^1 = \beta_i$ for $0\le i \le n_1$ and for $i \ge n_1 + 1$,
\begin{equation*}
\beta_i^1 = (1-\gamma_i) \beta_{i-1}^1 - \gamma_i \xi_i,
\end{equation*}
where $\xi_i$ is defined the same way as in \eqref{xi-def}, which is an mds. Let $\delta_i = \beta_i - \beta_i^1$ denote the approximation error. In the same manner as in the proof of Lemma \ref{lem:uniform approximation by beta_n^1}, we have 
\begin{equation*}
\sup_{1\le m \le n} \left\|\sum_{i=1}^m \delta_i \right\| \lesssim \sum_{i=n_1+1}^n \|((\Phi_{i-1}' W_{i-1} \Phi_{i-1})^{\dagger}\Phi_{i-1}' W_{i-1} G - I)\beta_{i-1}\| =: U_n.
\end{equation*}
Note that on an event $\{ T_R \ge i, W_{[i]} \in \mathcal{E}(R) \} \in \mathcal{F}_{i-1}$, it holds that $\|W_{i-1}\|\le R$ and $\|(\Phi_{i-1}' W_{i-1} \Phi_{i-1})^{\dagger}\| \le \|W_{i-1}^{-1}\| \|(\Phi_{i-1}\Phi_{i-1})^{\dagger}\| \le R^2$, and that $(\Phi_{i-1}' W_{i-1} \Phi_{i-1})^{\dagger}= (\Phi_{i-1}' W_{i-1} \Phi_{i-1})^{-1}$ for sufficiently large $i$ where $T_R = \min\{ \rho_R, \iota_R \}$ as in Part 1. Thus, we have that 
\begin{align*}
   & \mathbb{E}[U_n 1\{ T_R \ge n, W_{[n]} \in \mathcal{E}(R) \}] \\
   \le &\ \sum_{i=n_1+1}^n \mathbb{E}\left[ \|(\Phi_{i-1}' W_{i-1} \Phi_{i-1})^{\dagger}\|^{1/2} \|W_{i-1}\|^{1/2} \|\Phi_{i-1} - G\| \|\beta_{i-1}\| 1\{ T_R \ge i, W_{[i]} \in \mathcal{E}(R) \} \right]\\
   \lesssim &\ \sum_{i=1}^ n \mathbb{E}\left[ \|\Phi_{i-1} - G\| \|\beta_{i-1}\| 1\{ T_R \ge i, W_{[i]} \in \mathcal{E}(R) \}\right] \\
   \le &\ \sum_{i=1}^ n \mathbb E [\|\Phi_{i-1} - G\|^ 2]^{1/2} \mathbb E [\|\beta_{i-1}\|^ 21\{T_R \ge i-1, W_{[i-1]} \in \mathcal{E}(R)\} ]^{1/2}\\
   \lesssim &\ \sum_{i=1}^ n \frac{\gamma_i^ {1/2}}{\sqrt i} = O(n^{1/2 - a/2}) = o(\sqrt n).
\end{align*}
 Since $\limsup_{n\to\infty}\mathbb{P}(U_n \ne U_n 1\{ T_R \ge n, W_{[n]} \in \mathcal{E}(R) \}) < \varepsilon$ for a sufficiently large $R > 1$ as proven in Part 1, it holds $U_n = o_P(\sqrt n)$, and thus
$$
\sup_{1\le m\le n} \left\|\sum_{i=1}^ m \delta_i \right\| = o_P(\sqrt n).
$$
This, in particular, allows us to prove the CLT for $\bar \beta_{n}^{1}$ instead of $\bar \beta_{n}$. 

\noindent \textbf{\textit{Part 3}}. Establish the central limit theorem for $\sqrt{n}\bar \beta_n^1$. 

By construction of $\beta_i^1$, it holds $\sqrt{n}\bar \beta_n^1 = \sqrt{\frac{n_1}{n}} (\sqrt{n_1} \bar \beta_{n_1}) + \frac{1}{\sqrt n} \sum_{i= n_1 + 1}^ n \beta_i^1$. Since $\sqrt{n_1} \bar\beta_{n_1} = O_P(1)$ by Theorem \ref{thm:1} and $n_1 / n \to 0$, it suffices to establish the CLT for $\frac{1}{\sqrt n} \sum_{i= n_1}^ n \beta_i^1$. 

To this end, first observe that 
\begin{align*}
  \frac{1}{\sqrt n}  \sum_{i = n_1}^ n \beta_i^ 1 
  &= \frac{1}{\sqrt {n} \gamma_{n_1}} {\alpha}_{n_1}^n \beta_{n_1} - \frac{1}{\sqrt n}\sum_{i=n_1 + 1}^ n {\alpha}_i^ n \xi_i \\
  &= \frac{1}{\sqrt {n} \gamma_{n_1}} {\alpha}_{n_1}^n \beta_{n_1} - \frac{1}{\sqrt n}\sum_{i=n_1 + 1}^ n \xi_i  - \frac{1}{\sqrt n}\sum_{i=n_1 + 1}^ n {w}_i^ n \xi_i \\
  &= : I_1 - I_2 - I_3
\end{align*}
where $\alpha_j^n := \gamma_j \sum_{i=j}^ n \prod_{k = j+1}^i (1-\gamma_k)$ for all $j\le n$ and $w_j^n = \alpha_j^n - 1$.

For $I_1$, we know that $I_1 = o_P(1)$ because $|{\alpha}_{n_1}^n|$ is uniformly bounded in $n$ and $\frac{1}{\sqrt{n} \gamma_{n_1}}\beta_{n_1} = \frac{\sqrt{n_1}}{\sqrt{n}} n_1^{a-1/2} \beta_{n_1} = o_P(1)$ by Theorem~ \ref{thm:1}.

For $I_2$, we can apply the CLT (see \cite{Hall-Heyde}) for a triangular martingale difference array defined as  $\tilde{\xi}_{in} = n^{-1/2} \xi_i 1\{n_1 < i\le n\}$. This necessitates a version Lemma \ref{lem:Assumptions 2.3, 2.4, 2.5-(a) from PJ (1992)} as sufficient conditions for CLT, which will be presented in Part 2 of the proof of Theorem \ref{thm:4}. This yields $I_2 \overset{d}{\to} N(0, (G' \Omega^{-1} G)^{-1})$.

For $I_3$, we note that on an event $\{ T_R \ge n, W_{[n]} \in \mathcal{E}(R) \}$, it holds
\begin{equation*}
    I_3 = I_3' :=  \frac{1}{\sqrt n}\sum_{i=n_1}^ n {w}_i^ n \xi_i 1\{ T_R \ge i, W_{[i]} \in \mathcal{E}(R) \}.
\end{equation*}
Since the probability of this event can be made arbitrarily close to 1, it suffices to show that $I_3'$ is $o_P(1)$. Since $\|(\Phi_{i-1}' W_{i-1} \Phi_{i-1})^{\dagger}\| \le R^2$ and $\|W_{i-1}\|\le R$ on $\{ T_R \ge i, W_{[i]} \in \mathcal{E}(R) \}$, it holds that 
\begin{align*}
    \lambda_{\max}(\var(I_3')) &\le \frac{1}{n} \sum_{i=n_1+1}^n |w_i^n|^2 \mathbb{E}[\|\xi_i\|^2 1\{ T_R \ge i, W_{[i]} \in \mathcal{E}(R) \} ]\\
    &\lesssim \frac{1}{n} \sum_{i=n_1 + 1}^n |w_i^n|^2 \mathbb{E}[ \|\tilde{G}_i\|^2 \|\beta_{i-1}\|^2 1 \{ T_R \ge i, W_{[i]} \in \mathcal{E}(R) \}+ \|H_i\|^2  ] \\
    & = \frac{1}{n} \sum_{i=n_1 + 1}^n |w_i^n|^2 (\mathbb{E}[ \|\tilde{G}_i\|^2]\mathbb{E}[\|\beta_{i-1}\|^2 1 \{ T_R \ge i, W_{[i]} \in \mathcal{E}(R) \}]+ \mathbb{E}[\|H_i\|^2]) \\
    & \lesssim \frac{1}{n} \sum_{i=1}^n |w_i^n| (\gamma_i + 1)
     = O(n^{a - 1}) = o(1)
\end{align*}
where $\tilde{G}_i := G_i - G$. This shows $I_3 = o_P(1)$. 
 
We conclude that $\sqrt n \bar\beta_n$ converges in distribution to $ N(0, (G' \Omega^{-1} G)^{-1})$.

\end{proof}

\begin{proof}[\textbf{Proof of Theorem~\ref{thm:4}}]
\phantom{}\\
\noindent\textit{\textbf{Part 1}}.

Let $\bar \nu(r)$ and $\bar \nu^1(r)$ be defined as in \eqref{eq:define-nu}.
As demonstrated in the proof of Theorem \ref{thm:3}, $\sup_{0\le r\le 1}\|\bar\nu(r)-\bar\nu^1(r)\| = \frac{1}{\sqrt {n}}\sup_{1\le m\le n} \|\sum_{i=1}^m (\beta_i - \beta_i^1) \| = o_P(1)$ holds. As such, it is sufficient to prove the FCLT for $\bar\nu^1(r)$ in place of $\bar\nu(r)$. 

Further, note that 
$$
\sup_{0\le r\le n_1/n}\|\bar \nu(r)\|=\sup_{0\le r\le n_1/n}\|\bar \nu^1(r)\| = \sqrt{\frac{n_1}{n}} \cdot \sup_{0\le s\le 1}\frac{1}{\sqrt {n_1}}\left\|\sum_{i=1}^{\lfloor n_1 s \rfloor} \beta_i\right\| = o_P(1)
$$
in light of Theorem \ref{thm:2} and $\frac{n_1}{n}=o(1)$. Thus, we may focus on $r \in [n_1/n, 1]$, which allows us to consider the following decomposition.

Analogous to the proof of Theorem \ref{thm:2}, we have 
\begin{equation*}
    \bar \nu^1(r) - \bar \nu^1(n_1/n) = I_1(r) - I_2(r) - I_3(r)
\end{equation*}
where 
\begin{align*}
    I_1(r) &= \frac{1}{\sqrt{n}\gamma_{n_1}}\alpha_{n_1}^{\lfloor nr \rfloor} \beta_{n_1}, \\
    I_2(r) &= \frac{1}{\sqrt n}\sum_{i=n_1 + 1}^{\lfloor nr \rfloor} \xi_i, \\
    I_3(r) &= \frac{1}{\sqrt n}\sum_{i=n_1 + 1}^{\lfloor nr \rfloor} w_i^{\lfloor nr \rfloor} \xi_i.    
\end{align*}

It is easy to see that the first term is $\sup_{r \in [0,1]}\|I_1(r)\| \lesssim n_1^{a}\|\beta_{n_1}\|/\sqrt {n} = o_P(1)$ because $|\alpha_{n_1}^{\lfloor nr \rfloor}|$ is uniformly bounded from above in $n$ and $r \ge n_1/n$.

The second term is where the FCLT for a triangular martingale array applies and yields 
\begin{equation*}
    \{I_2(r)\}_{r \in [0,1]} \rightsquigarrow (G'\Omega^{-1} G)^{-1/2} \{ W_{d_\beta}(r) \}_{r \in [0, 1]}.
\end{equation*}
We will verify the sufficient conditions in \textbf{Part 2} of this proof.

It remains to show the third term $\sup_{r \in [0,1]}\|I_3(r)\| = o_P(1)$.
Let $S_m = \sum_{i=n_1+1}^ m w_i^ m \xi_i$ so that $\sup_{0\le r\le 1}\|I_3(r)\| = n^{-1/2} \sup_{n_1 < m\le n} \|S_m\|$ holds. Let $p$ be the integer that appears in Assumption \ref{A:moment:FCLT}. We note that on an event $\{ T_R \ge n, W_{[n]} \in \mathcal{E}(R) \}$, $\sup_{n_1 < m\le n}\|S_m\|^{2p} \le \sum_{m=n_1+1}^n \|S_m\|^{2p} = \sum_{m=n_1+1}^n \|\tilde{S}_m\|^{2p}$ where $\tilde{S}_m :=  \sum_{i=n_1+1}^ m w_i^ m \xi_i 1\{ T_R \ge i, W_{[i]} \in \mathcal{E}(R) \}$. Thus, by Burkholder's inequality (\cite{Hall-Heyde}) applied to $\|\tilde{S}_m\|^{2p}$, we get
\begin{align*}
  & \mathbb{E}[ \sup_{r\in[0,1]} \|I_3(r)\|^{2p} 1\{ T_R \ge n, W_{[n]} \in \mathcal{E}(R) \} ] 
  \le \ n^{-p}\sum_{m=n_1 + 1}^ n \mathbb{E}[\|\tilde{S}_m\|^{2p}] \\
  \lesssim & \ n^{-p} \sum_{m=n_1 + 1}^ n \mathbb{E}\left( \sum_{i=n_1 + 1}^m |w_i^m|^2 \|\xi_i\|^2 1\{ T_R \ge i, W_{[i]} \in \mathcal{E}(R) \} \right)^p \\
  \lesssim & \ n^{-p} \sum_{m=n_1 + 1}^ n \left(\sum_{i=n_1 + 1}^m |w_i^m|^2\right)^p \mathbb{E}[ \|\xi_i\|^{2p} 1\{ T_R \ge i, W_{[i]} \in \mathcal{E}(R) \}] \\
  \lesssim & \ n^{-p} \sum_{m=n_1 + 1}^ n \left(\sum_{i=n_1 + 1}^m |w_i^m|\right)^p( \mathbb{E}[ \|\tilde{G}_i\|^{2p}]\mathbb{E}[\|\beta_{i-1}\|^{2p} 1\{ T_R \ge i, W_{[i]} \in \mathcal{E}(R)\}] + \mathbb{E}[\|H_i\|^{2p}])\\
  \lesssim & \ n^{-p} \sum_{m = n_1 + 1}^ n m^{ap} 
 \lesssim \ n^{-(1-a)p + 1} = o(1).
\end{align*}
Here, we used $\mathbb{E}[\|\beta_i\|^{2p} 1\{ T_R \ge i, W_{[i]} \in \mathcal{E}(C)\}] = O(\gamma_i) = O(1)$ ($p > 2$), which can be established with additional assumptions that $\mathbb{E}[\|\beta_0\|^{2p}] < \infty$ and Assumption \ref{A:moment:FCLT}. The proof follows the same approach as Part 1 of the proof of Theorem \ref{thm:3} and therefore will be omitted. This shows $\sup_{r\in [0, 1]} \|I_3(r)\|=o_P(1)$ since $\limsup_{n\to \infty}\mathbb{P}(\{ T_R \ge n, W_{[n]} \in \mathcal{E}(R) \}^c) < \varepsilon$ for a sufficiently large $R$. 

Putting all together, we conclude 
\[\{ \bar \nu(r) \}_{0\le r \le 1} \rightsquigarrow (G'\Omega^{-1}G)^{-1/2} \{ \mathbb{W}_{d_\beta}(r) \}_{0\le r\le 1}.\]

\noindent \textbf{\textit{Part 2}}. Lindeberg conditions for CLT and FCLT.

This section establishes: 
\begin{align}
    \label{eq:lindeberg-1}
    &\operatornamewithlimits{plim}_{n\to\infty}\sum_{i=1}^n \mathbb{E}[\|\tilde{\xi}_{in}\|^2 1\{ \|\tilde{\xi}_{in}\|\ge \varepsilon\} \mid \mathcal{F}_{i-1}] = 0,\quad\forall \varepsilon>0,\\
    \label{eq:lindeberg-2}
    &\operatornamewithlimits{plim}_{n\to\infty}\sum_{i=1}^n \mathbb{E}[\tilde{\xi}_{in} \tilde{\xi}_{in}' \mid \mathcal{F}_{i-1}] = (G'\Omega^{-1} G)^{-1}.
\end{align}
where 
$$
\tilde{\xi}_{in} := \frac{1}{\sqrt n}\xi_i 1\{ n_1 <  i\le n \}=\frac{1}{\sqrt n}(\Phi_{i-1}' W_{i-1} \Phi_{i-1})^{\dagger} \Phi_{i-1}'W_{i-1} (\tilde{G}_i \beta_{i-1} + H_i) 1\{ n_1 < i\le n \}.
$$

\noindent \textbf{\textit{Part 2-a}}.

For \eqref{eq:lindeberg-1}, we first note that it is sufficient to show that 
\begin{equation}
    \label{eq:lindeberg-1'}
   A_n := \sum_{i=1}^n \mathbb{E}[\|\tilde{\xi}_{in}\|^2 1\{ T_R \ge i, W_{[i]} \in \mathcal{E}(R) \} 1\{ \|\tilde{\xi}_{in}\|\ge \varepsilon\} \mid \mathcal{F}_{i-1}] \to 0
\end{equation}
in probability for each $R > 1$ because \eqref{eq:lindeberg-1} is equivalent to \eqref{eq:lindeberg-1'} on a set of arbitrarily large probability for a sufficiently large $R$. Since 
\begin{equation*}
    \tilde{\xi}_{in} \le \frac{1}{\sqrt n} R^{3/2} \|\tilde{G_i}\beta_{i-1} + H_i\| 1\{ n_1 < i\le n \}
\end{equation*}
on $\{ T_R \ge i, W_{[i]} \in \mathcal{E}(R) \}$, it follows by Markov's inequality
\begin{align*}
    A_n &\lesssim  \frac{1}{n} \sum_{i=n_1 + 1}^n \mathbb{E}[\|\tilde{G_i}\beta_{i-1} + H_i\|^2 1\{ \|\tilde{G_i}\beta_{i-1} + H_i\| \ge \sqrt n \varepsilon/R^2, T_R \ge i, W_{[i]} \in \mathcal{E}(R) \}\mid \mathcal{F}_{i-1}]\\
    &\lesssim \frac{1}{n^{p}} \sum_{i=n_1 + 1}^n \mathbb{E}[\|\tilde{G_i}\beta_{i-1} + H_i\|^{2p} 1\{ T_R \ge i, W_{[i]} \in \mathcal{E}(R) \}\mid \mathcal{F}_{i-1}].
\end{align*}
Taking expectations on both sides, 
\begin{align*}
    \mathbb{E}[A_n] &\le  \frac{1}{n^{p}} \sum_{i=n_1 + 1}^n \mathbb{E}[\|\tilde{G_i}\beta_{i-1} + H_i\|^{2p} 1\{ T_R \ge i, W_{[i]} \in \mathcal{E}(R) \}] \\
    &\lesssim \frac{1}{n^{p}} \sum_{i=n_1 + 1}^n \mathbb{E}[(\|\tilde{G_i}\|^{2p}\|\beta_{i-1}\|^{2p} + \|H_i\|^{2p}) 1\{ T_R \ge i, W_{[i]} \in \mathcal{E}(R) \}] \\
    &\le \frac{1}{n^{p}} \sum_{i=n_1 + 1}^n \mathbb{E}[\|\tilde{G_i}\|^{2p}]\mathbb{E}[\|\beta_{i-1}\|^{2p}1\{ T_R \ge i, W_{[i]} \in \mathcal{E}(R) \}] + \mathbb{E}[\|H_i\|^{2p}] \\
    &= O(n^{1-p}) = o(1),
\end{align*}
which follows from $\mathbb{E}[\|\beta_{i}\|^{2p}1\{ T_R \ge i, W_{[i]} \in \mathcal{E}(R) \}] = O(\gamma_i)=O(1)$. This proves $A_n \to 0$ in probability. 

\noindent \textbf{\textit{Part 2-b}}. Denote $\Omega_\beta := \mathbb{E}[g_i(\beta)g_i(\beta)']$. 
For \eqref{eq:lindeberg-2}, we write 
\begin{align*}
    \sum_{i=1}^n \mathbb{E}[\tilde{\xi}_{in} \tilde{\xi}_{in}' \mid \mathcal{F}_{i-1}] &= \frac{1}{n}\sum_{i=n_1 + 1}^n \mathbb{E}[\xi_i\xi_i' \mid \mathcal{F}_{i-1}] \\
    &= \frac{1}{n}\sum_{i=n_1 + 1}^n (\Phi_{i-1}' W_{i-1} \Phi_{i-1})^{\dagger}\Phi_{i-1}'W_{i-1} F(\beta_{i-1}) W_{i-1}\Phi_{i-1} (\Phi_{i-1}' W_{i-1} \Phi_{i-1})^{\dagger} \\
    & \ \ + \frac{1}{n}\sum_{i=n_1 + 1}^n (\Phi_{i-1}' W_{i-1} \Phi_{i-1})^{\dagger}\Phi_{i-1}'W_{i-1} \Omega W_{i-1}\Phi_{i-1} (\Phi_{i-1}' W_{i-1} \Phi_{i-1})^{\dagger}\\
    &=: D_1 + D_2,
\end{align*}
where $F(\beta) = \Omega_\beta - \Omega = \mathbb{E}[G_i\beta \beta'G_i] + \mathbb{E}[G_i \beta H_i']+\mathbb{E}[H_i \beta G_i']$ and $\Omega = \mathbb{E}[H_iH_i']$. To show $D_1 = o_P(1)$, we shall establish instead 
\begin{align*}
    &\frac{1}{n}\sum_{i=n_1 + 1}^n (\Phi_{i-1}' W_{i-1} \Phi_{i-1})^{\dagger}\Phi_{i-1}'W_{i-1} F(\beta_{i-1}) W_{i-1}\Phi_{i-1} (\Phi_{i-1}' W_{i-1} \Phi_{i-1})^{\dagger}1\{ T_R \ge i, W_{[i]} \in \mathcal{E}(R) \}\\
    &= o_P(1).
\end{align*}
Using the fact that $\|F(\beta)\| \lesssim \|\beta\|^2 + \|\beta\|$, we get 
\begin{align*}
    &\frac{1}{n}\left\|\sum_{i=n_1 + 1}^n (\Phi_{i-1}' W_{i-1} \Phi_{i-1})^{\dagger}\Phi_{i-1}'W_{i-1} F(\beta_{i-1}) W_{i-1}\Phi_{i-1} (\Phi_{i-1}' W_{i-1} \Phi_{i-1})^{\dagger}1\{ T_R \ge i, W_{[i]} \in \mathcal{E}(R) \}\right\|\\
    \lesssim &\ \frac{1}{n}\sum_{i=n_1 + 1}^n (\|\beta_{i-1}\|^2 + \|\beta_{i-1}\|) 1\{ T_R \ge i, W_{[i]} \in \mathcal{E}(R) \}.
\end{align*}
Since $\frac{1}{n}\sum_{i=n_1 + 1}^n \mathbb{E}[(\|\beta_{i-1}\|^2 + \|\beta_{i-1}\|) 1\{ T_R \ge i, W_{[i]} \in \mathcal{E}(R) \}] \lesssim \frac{1}{n} \sum_{i=1}^n \gamma_i^{1/2} = o(1)$, it follows $D_1 = o_P(1)$.

Next, we establish $D_2 \to (G'\Omega^{-1}G)^{-1}$ in probability. For $\beta \in \mathcal{K}$, let us define
\[
V(\beta) := (G' \Omega_{\beta}^{-1} G)^{-1}G'\Omega_{\beta}^{-1} \Omega \Omega_{\beta}^{-1} G(G'\Omega_{\beta}^{-1} G)^{-1}.
\]
and for simplicity of notation, we denote
\[
V_{i-1} = F(\Phi_{i-1}, W_{i-1}^{-1}) := (\Phi_{i-1}' W_{i-1} \Phi_{i-1})^{\dagger}\Phi_{i-1}'W_{i-1}\Omega W_{i-1}\Phi_{i-1}(\Phi_{i-1}' W_{i-1} \Phi_{i-1})^{\dagger}
\]
so that $D_2 = \frac{1}{n}\sum_{i= n_1 + 1}^n V_{i-1}$. 
Given this, we aim to show that
\begin{equation*}
\frac{1}{n} \sum_{i= n_1 + 1}^ n \|V_{i-1} - V(0)\| = o_P(1),
\end{equation*}
where $V(0) = (G'\Omega^{-1}G)^{-1}$. To this end, we note it is sufficient to prove that
\begin{equation}
    \label{eq:obj1}
   \frac{1}{n} \sum_{i= n_1 + 1}^ n \|V_{i-1} - V(0)\|1\{ T_R \ge i, W_{[i]} \in \mathcal{E}(R), \bar\beta_{n_1} \in \mathcal{K}\} = o_P(1)
\end{equation}
for each $R$. 

We see that $V_{i-1}=F(\Phi_{i-1}, W_{i-1}^{-1})$ is locally Lipschitz continuous on a neighborhood of $(G, \Omega)$. As such, it follows that on the event $\{ T_R \ge i, W_{[i]} \in \mathcal{E}(R), \bar\beta_{n_1} \in \mathcal{K} \}$, if $\|\Phi_{i-1} - G\| < \delta$ and $\|W_{i-1}^{-1} - \Omega\| < \delta$ for a sufficiently small $\delta > 0$,
\begin{align*}
    \|V_{i-1} - V(0)\| &= \|F(\Phi_{i-1}, W_{i-1}^{-1}) - F(G, \Omega) \|\\
    &\lesssim \|\Phi_{i-1} - G\| + \|W_{i-1}^{-1} - \Omega\|,
\end{align*}
and otherwise, $\|V_{i-1} - V(0)\| \lesssim 1 \lesssim  1\{ \|\Phi_{i-1} - G\| \ge \delta \} + 1\{\|W_{i-1}^{-1} - \Omega\| \ge \delta  \} \lesssim \|\Phi_{i-1} - G\| + \|W_{i-1}^{-1} - \Omega\|$. This observation allows us to prove that,  instead of \eqref{eq:obj1}, 
\begin{align*}
   & \left( \frac{1}{n} \sum_{i=n_1 + 1}^{n}\|\Phi_{i-1} - G\| + \|W_{i-1}^{-1} - \Omega\|1\{ T_R \ge i, W_{[i]} \in \mathcal{E}(R), \bar\beta_{n_1} \in \mathcal{K} \}\right)^2 \\
   \lesssim & \ \frac{1}{n} \sum_{i=n_1 + 1}^{n} \|\Phi_{i-1} - G\|^2 + \frac{1}{n}\sum_{i = n_1 + 1}^n \|W_{i-1}^{-1} - \Omega\|^2 1\{ T_R \ge i, W_{[i]} \in \mathcal{E}(R) , \bar\beta_{n_1} \in \mathcal{K}\} \\
   =: &\ A_1 + A_2
\end{align*}
is $o_P(1)$. The assertion that $A_1 = o_P(1)$ follows from 
\begin{equation*}
    \mathbb{E}[A_1] = \frac{1}{n}\sum_{i=n_1 + 1}^n \mathbb{E}[\|\Phi_{i-1} - G\|^2] \lesssim \frac{1}{n}\sum_{i=1}^n \frac{1}{i} = o(1). 
\end{equation*}
For $A_2$, we first note that on $\{ T_R \ge i, W_{[i]} \in \mathcal{E}(R), \bar\beta_{n_1} \in \mathcal{K}\}$, it holds 
\begin{align*}
    \|W_{i}^{-1} - \Omega\| &= \left\|\frac{n_0 + n_1}{n_0 + i}(W_{n_1}^{-1} - \Omega) + \frac{1}{n_0+i}\sum_{j=n_1 + 1}^{i} (g_j(\bar\beta_{n_1})g_j(\bar\beta_{n_1})' - \Omega)\right\|  \\
    &\lesssim \frac{n_0 + n_1}{n_0 + i} + \frac{1}{i}\left\| \sum_{j=n_1 + 1}^{i} (g_j(\bar\beta_{n_1})g_j(\bar\beta_{n_1})' - \Omega_{\bar\beta_{n_1}})\right\| + \|\Omega_{\bar\beta_{n_1}} - \Omega\|.
\end{align*}
This gives $\|W_{i}^{-1} -\Omega\|^2 \lesssim (\frac{n_0 + n_1}{n_0 + i})^2 + (\frac{1}{i^2}  \|\sum_{j=n_1 + 1}^{i} (g_j(\bar\beta_{n_1})g_j(\bar\beta_{n_1})' - \Omega_{\bar\beta_{n_1}})\|^2 + \|\Omega_{\bar\beta_{n_1}} - \Omega\|^2) 1\{ \bar\beta_{n_1} \in \mathcal{K} \}$ on this event, hence 
\begin{align*}
\mathbb{E}[A_2] &\lesssim \frac{1}{n} \sum_{i=n_1 + 1}^n  \left(\frac{n_0 + n_1}{n_0 + i}\right)^2 +   \frac{1}{n} \sum_{i=n_1 + 1}^n \frac{1}{i^2}\mathbb{E}\left[ \left\|\sum_{j=n_1 + 1}^{i} (g_j(\bar\beta_{n_1})g_j(\bar\beta_{n_1})' - \Omega_{\bar\beta_{n_1}})\right\|^21\{ \bar\beta_{n_1} \in \mathcal{K} \}\right]\\
& \quad + \mathbb{E}[\|\Omega_{\bar\beta_{n_1}} - \Omega\|^2 1\{ \bar\beta_{n_1} \in \mathcal{K} \}]
\end{align*}
The first term can be estimated as $O(\frac{n_1}{n}) = o(1)$. For the second term, we note, by the law of iterated expectations,
\begin{align*}
   & \mathbb{E}\left[ \left\|\sum_{j=n_1 + 1}^{i} (g_j(\bar\beta_{n_1})g_j(\bar\beta_{n_1})' - \Omega_{\bar\beta_{n_1}})\right\|^21\{ \bar\beta_{n_1} \in \mathcal{K} \}\right] \\
   = & \ \int_{\mathcal{K}} \mathbb{E}\left[\left\|\sum_{j=1}^{i-n_1} (g_j(\beta)g_j(\beta)' - \Omega_{\beta})\right\|^2\right]d\mu_{n_1}(\beta) \\
   \le & \ \sup_{\beta \in \mathcal{K}}   \mathbb{E}\left[\left\|\sum_{j=1}^{i-n_1} (g_j(\beta)g_j(\beta)' - \Omega_{\beta})\right\|^2\right]    \lesssim \ i
\end{align*}
where $\mu_{n_1}(\cdot)$ represents the distribution of $\bar\beta_{n_1}$ as before. This establishes that the second term is bounded by $O(\frac{1}{n}\sum_{i=1}^n \frac{1}{i}) = o(1)$. For the last term, since $\|\bar\beta_{n_1}\| = o_P(1)$, it is $o(1)$ by the dominated convergence theorem. We conclude $D_2$ converges in probability to $V(0) = (G'\Omega^{-1}G)^{-1}$.
\end{proof}


\begin{proof}[\textbf{Proof of Corollary~\ref{cor:dwh}}]
It follows from the same proof of Theorem \ref{thm:2} that 
 \begin{equation*}
        \left\{\frac 1 {\sqrt {n}}  \sum_{i=1}^{\lfloor n r\rfloor}(\bbeta_i - \bbeta_*) \right\}_{r\in[0,1]}\rightsquigarrow \Gamma^{1/2}\{ W_{2d_\beta}(r)\}_{r\in [0,1]},
    \end{equation*}  
where  $\Gamma=(\bG' \bW \bG)^{-1} \bG' \bW \bOmega \bW \bG   (\bG' \bW\bG)^{-1}  $,  $\bOmega$ is the covariance matrix of $(z_i'u_i, x_i'u_i)'$, $\bG=\mathbb E\bG_i$,
and
$$
\bW= \begin{pmatrix}
\{ \mathbb{E}[z_i z_i']  \}^{-1} &0\\
0& I
\end{pmatrix}.  
$$
 The proof of this FLCT  is very similar to that of  Theorem \ref{thm:2}. Hence we omit the details for brevity. Also, there is a selection matrix $S$ such that we can express
 $$
 \bar\beta_{\sub,n}-\bar\alpha_{\sub, n}= \Xi \bar\bbeta_{\sub,n}= S\bar\bbeta_n.
 $$

Let $C_n(r):= S\frac 1 {\sqrt {n}}  \sum_{i=1}^{\lfloor n r\rfloor}(\bbeta_i - \bbeta_*) $. Also let $\Lambda:=( S\Gamma S')^{1/2}$.  Then $C_n(r)$ weakly converges to $\Lambda \mathbb{W}_{q}(r)$ where $q$ is the rank of $S$.  Under the null that $\alpha_{\sub}=\beta_{\sup}$ we have $S\bbeta_*=0
.  $ Hence 
$\bar\beta_{\sub,n}-\bar\alpha_{\sub, n}= \frac{1}{\sqrt{n}} C_n(1)$. Also, 
\begin{eqnarray*}
\Xi\widehat{V}_{\sub,n}\Xi'&=&S\widehat{V}_{rs,n}S'=\frac{1}{n}\sum_{s=1}^n[C_n(\frac{s}{n})- \frac{s}{n}C_n(1)][C_n(\frac{s}{n})- \frac{s}{n}C_n(1)]'\cr
&=&\int_0^1[C_n(r)-rC_n(1)][C_n(r)-rC_n(1)]'dr,
\end{eqnarray*}
the last equality holds because $C_n(r)$ is a partial sum. Hence 
$$\|n(\bar\beta_{\sub,n}-\bar\alpha_{\sub, n})'(\Xi\widehat{V}_{\sub,n}\Xi')^{-1} (\bar\beta_{\sub,n}-\bar\alpha_{\sub, n})\|$$ can be expressed as a continuous functional of $C_n(r)$. The result then follows from the continuous mapping theorem.
\end{proof}

\begin{proof}[\textbf{Proof of Corollary~\ref{overid:test}}]
It follows directly from Theorem~\ref{thm:3}.
\end{proof}

\subsection{Auxiliary Lemmas}
\begin{lem}[Robbins-Siegmund]
\label{lem:Robbins Siegmund}
Suppose that $Z_n, A_n, C_n$, and $D_n$ are finite, non-negative random variables, adapted to the filtration $\{ \mathcal{F}_n \}_{n =0}^\infty$, which satisfy
\begin{equation*}
    \mathbb E [Z_{n+1}|\mathcal{F}_n] \le (1 + A_n) Z_n + C_n - D_n.
\end{equation*}
Then, on the event $\{\sum_{n=1}^\infty A_n < \infty,\ \sum_{n=1}^\infty C_n < \infty\}$, we have 
\begin{equation*}
    \sum_{n=1}^\infty D_n <\infty \quad  \text{ and } \quad Z_n \to Z
\end{equation*}
almost surely, where $Z$ denotes the limiting random variable $\lim_{n\to\infty} Z_n < \infty$.
\end{lem}
\begin{proof}
See Lemma 5.2.2 in \citet[p. 344]{benveniste2012adaptive}.
\end{proof}

Lemma \ref{lem:recursive bound} is akin to Theorem 24 in \citet[pp. 246--247]{benveniste2012adaptive}, but 
we present a self-contained proof for the sake of completeness.

\begin{lem} 
\label{lem:recursive bound}
Assume that $\gamma_n = \gamma_0 n^{-a}$ for $a \in (1/2, 1)$ and $\gamma_0 > 0$. Let $a_n$ be non-negative numbers such that 
\begin{equation*}
    a_{n+1} \le (1-\alpha \gamma_n) a_n + C \gamma_n^ {q+1}
\end{equation*}
for all sufficiently large $n$ and some positive numbers $\alpha, C$ and $q \ge 0$. Then, $a_n = O(\gamma_n^q)$ holds.
\end{lem}
\begin{proof}

Let $\delta_n := a_n / \gamma_n^q$ and we show $\limsup_{n\to \infty} \delta_n < \infty$ by way of contradiction. Since $\gamma_n - \gamma_{n+1} = o(\gamma_n^ 2)$, for a sufficiently small $\varepsilon_1 > 0$, $\gamma_{n+1} \ge \gamma_n - \varepsilon_1 \gamma_n^ 2$ for all large $n$. Replacing $a_n$ with $\gamma_n^q \delta_n$ and using $(\gamma_n / \gamma_{n+1})^q \le (1-\varepsilon_1 \gamma_n)^{-q}\le 1+\varepsilon_2 \gamma_n$ for an arbitrarily small $\varepsilon_2 = O(\varepsilon_1)$, we have 
\begin{equation*}
    \delta_{n+1} \le (1-\alpha \gamma_n) \delta_n + C\gamma_n
\end{equation*}
for all sufficiently large $n$ after some relabeling of $\alpha$ and $C$. This implies that
\begin{equation}
    \label{eq:lemma-3}
    \delta_{n+1} - \delta_n \le   \gamma_n (-\alpha\delta_n + C).
\end{equation}
Assume to the contrary $\limsup_n \delta_n = \infty$. Choose $N$ such that $\gamma_N < 1/\alpha$. If $\delta_n > 2C/\alpha$ holds for all $n \ge N$, it leads to a contradiction because $\lim_{n \to \infty}\delta_n  = \sum_{n\ge N} (\delta_{n+1} - \delta_n) +\delta_N \le \sum_{n\ge N} \gamma_n(-\alpha \delta_n + C) +\delta_N\le - C \sum_{n\ge N} \gamma_n +\delta_N =-\infty$. Thus, we can choose $N' \ge N$ such that $\delta_{N'} \le 2C/\alpha$. On the other hand, $\limsup_n \delta_n =\infty > 2C/\alpha$ precludes the possibility that $\delta_n \le 2C/\alpha$ for all $n\ge N'$. Thus, there must exist $n_0 \ge N'$ such that $\delta_{n_0} \le 2C/\alpha < \delta_{n_0+1}$. Due to the fact that $\delta_{n_0+1} - \delta_{n_0} > 0$ and \eqref{eq:lemma-3}, it must hold $\delta_{n_0} < C /\alpha$. However, this leads to a contradiction since $\delta_{n_0+1} - \delta_{n_0}\le \gamma_n(-\alpha \delta_n + C) \le C \gamma_{n_0} < C/\alpha$ by the choice of $N$, but $\delta_{n_0+1} > 2C/\alpha$ and $\delta_{n_0} < C/\alpha$. Therefore, $\limsup_{n\to\infty} \delta_n \le 2C/\alpha < \infty$ under \eqref{eq:lemma-3}.

\end{proof}

\begin{lem}
\label{lem:Assumptions 2.3, 2.4, 2.5-(a) from PJ (1992)}
Let Assumptions~\ref{A:sample}--\ref{A:moment:consistency} hold. Assume that $W_i$ obeys \eqref{algo:beta:W:2sls}. Then, the following holds.

\begin{enumerate}[leftmargin=0.7cm]

    \item [(a)] $\sup_{i \in \mathbb N} \mathbb E[\|\xi_i\| ^2 \mid \mathcal F_{i-1}] <\infty$ almost surely.
    \item [(b)] $\sup_{i \in \mathbb N} \mathbb E [\|\xi_i\|^ 2 1\{\|\xi_i\| > C\}\mid \mathcal F_{i-1}] = o_{a.s.}(1)$ as $C \to \infty$.
    \item [(c)] $\lim_{i\to\infty} \mathbb E[ \xi_i\xi_i' \mid \mathcal F_{i-1}] = (G'WG)^{-1}G'W \Omega W G(G'WG)^{-1}$ almost surely for $\Omega = \var(g_i(\beta_*))$. 

\end{enumerate}
\end{lem}
\begin{proof} We maintain the assumption that $\beta_* = 0$. 
 
\noindent (a) Recalling that $\|(\Phi_{i-1}' W_{i-1} \Phi_{i-1})^{\dagger} \Phi_{i-1}' W_{i-1}^{1/2} \| = \|(\Phi_{i-1}' W_{i-1} \Phi_{i-1})^{\dagger}\|^{1/2}$, we see $\|\xi_i\|^ 2 \le  \|W_{i-1}\| \|(\Phi_{i-1}' W_{i-1} \Phi_{i-1})^{\dagger}\|(\|{\tilde G_i}\| \|\beta_{i-1}\| + \|H_i\|)^ 2$. Using $(a+b)^ 2 \le 2(a^ 2 + b^ 2)$, we obtain
\begin{equation*}
    \mathbb E[\|\xi_i\|^ 2 \mid \mathcal F_{i-1}] \lesssim\|W_{i-1}\| \|(\Phi_{i-1}' W_{i-1} \Phi_{i-1})^{\dagger}\| (\mathbb E[\|{\tilde G_i}\|^ 2] \|\beta_{i-1}\|^ 2 + \mathbb E[\|H_i\|]^ 2) = O_{a.s.}(1)
\end{equation*}
where we used $\sup_{i\ge 1}\max\{ \|W_{i-1}\|, \|(\Phi_{i-1}' W_{i-1} \Phi_{i-1})^{\dagger}\|, \|\beta_{i-1}\|  \} < \infty$ implied by the SLLN and Lemma \ref{lem:strong consistency of beta_n}. 

\noindent (b) 
It is not proven in the same manner as Theorem 2 in \cite{polyak1992acceleration}, because it incorporates a stopping time $T_R := \min\{\tau_R, \sigma_R\}$ to address the randomness in $(\Phi_{i-1}' W_{i-1} \Phi_{i-1})^{\dagger}$, $W_{i-1}$, and $\beta_{i-1}$. Let $R > 1$.

Since $\|(\Phi_{i-1}' W_{i-1} \Phi_{i-1})^{\dagger}\|\le R$, $\|W_{i-1}\| \le R$, and $\|\beta_{i-1}\| \le R$ on the event $\{ T_R \ge i \} \in \mathcal{F}_{i-1}$, we have
\begin{align*}
    &\mathbb E [\|\xi_i\|^ 2 1\{\|\xi_i\| > C\}\mid \mathcal F_{i-1}]\\
    =&\ 1\{ T_R \ge i \}\mathbb E [\|\xi_i\|^ 2 1\{\|\xi_i\| > C\}\mid \mathcal F_{i-1}] + 1\{ T < i \}\mathbb E [\|\xi_i\|^ 2 1\{\|\xi_i\| > C\}\mid \mathcal F_{i-1}] \\
    \le &\ 1\{ T_R \ge i \} \mathbb E [R^4(\|\tilde{G}_i\| + \|H_i\|)^2 1\{\|\tilde{G}_i\| + \|H_i\| > C/R^2 \}\mid \mathcal F_{i-1}] \\
    & \quad + 1\{ T_R < \infty \} \sup_{i\ge 1}\mathbb E [\|\xi_i\|^ 2 1\{\|\xi_i\| > C\}\mid \mathcal F_{i-1}] \\
    \le &\ R^4\mathbb E [(\|\tilde{G}_i\| + \|H_i\|)^2 1\{\|\tilde{G}_i\| + \|H_i\| > C/R^2 \}] + 1\{ T_R < \infty \} \sup_{i\ge 1}\mathbb E [\|\xi_i\|^ 2 1\{\|\xi_i\| > C\}\mid \mathcal F_{i-1}]
\end{align*}
where $\sup_{i\in\mathbb N}\mathbb E[\|\xi_i\|^ 2 \mid \mathcal F_{i-1}] <\infty$ comes from part (a). Note that the rightmost side does not depend on $i$ anymore. As such, it serves as a uniform bound of the leftmost side across all $i$ for any given $R > 1$. Taking limit superior as $C\to \infty$, it follows
\begin{equation*}
    \limsup_{C\to\infty}\sup_{i\in\mathbb N} \mathbb E[\|\xi_i\|^ 2 1\{\|\xi_i\| > C\} \mid \mathcal F_{i-1}]  \le 1\{T_R <\infty\}\sup_{i\in\mathbb N}\mathbb E[\|\xi_i\|^ 2 \mid \mathcal F_{i-1}].
\end{equation*}
Finally, the right-hand side can be made 0 by increasing $R$ since $1\{T_R < \infty\} = 0$ for all sufficiently large $R$ with probability 1. We conclude that
\begin{equation*}
    \lim_{C\to\infty}\sup_{i\in\mathbb N} \mathbb E[\|\xi_i\|^ 2 1\{\|\xi_i\| > C\} \mid \mathcal F_{i-1}] = 0   
\end{equation*}
almost surely. 

\noindent (c) Observe that
\begin{align*}
        \xi_i\xi_i' &= (\Phi_{i-1}' W_{i-1} \Phi_{i-1})^{\dagger}\Phi_{i-1}'W_{i-1}(\tilde G_i \beta_{i-1}\beta_{i-1}' \tilde G_i' + \tilde G_i \beta_{i-1} H_i' + H_i \beta_{i-1}'\tilde G_i' + H_i H_i')\\
        & \quad \times W_{i-1}\Phi_{i-1} (\Phi_{i-1}' W_{i-1} \Phi_{i-1})^{\dagger}.
\end{align*}
By the strong consistency of $\beta_{n}$ established in Lemma \ref{lem:strong consistency of beta_n}, we have 
\begin{align*}
    \mathbb E[ \xi_i\xi_i' \mid \mathcal F_{i-1}] &= (\Phi_{i-1}' W_{i-1} \Phi_{i-1})^{\dagger} \Phi_{i-1}' W_{i-1}\mathbb E[H_i H_i'] W_{i-1}\Phi_{i-1} (\Phi_{i-1}' W_{i-1} \Phi_{i-1})^{\dagger} + o_{a.s.}(1) \\
    & \to (G'WG)^{-1}G'W \Omega W G (G'WG)^{-1}.
\end{align*}
as $i \to \infty$ as a result of the SLLN applied to $W_{i-1}$ and $\Phi_{i-1}$.
\end{proof}

\begin{lem}
\label{lem:uniform approximation by beta_n^1}
Let Assumptions~\ref{A:sample}--\ref{A:moment:consistency} hold.
Define $\beta_n^ 1$ as in \eqref{eq:definition of beta_n^1:2sls:sgmm}. Then it holds 
\begin{equation*}
    \sup_{1\le m\le n} \left\|\sum_{i=1}^m (\beta_i - \beta_i^ 1)\right\|= o_P(\sqrt n).    
\end{equation*}
\end{lem}
\begin{proof}
Let $\delta_i = \beta_i - \beta_i^ 1$. By construction, it holds that 
\begin{align*}
    \delta_i &= (1 - \gamma_i )\delta_{i-1} + \gamma_i\left[I- (\Phi_{i-1}' W_{i-1} \Phi_{i-1})^{\dagger} \Phi_{i-1}' W_{i-1} G\right] \beta_{i-1}  \\
    &=  (1 - \gamma_i )\delta_{i-1} + \gamma_i \kappa_i
\end{align*}
where $\kappa_i := \left[I- (\Phi_{i-1}' W_{i-1} \Phi_{i-1})^{\dagger} \Phi_{i-1}' W_{i-1} G\right] \beta_{i-1}$. This readily implies that 
\begin{equation*}
    \sum_{i=1}^ m \delta_i = \sum_{i=1}^ m \alpha_{i}^ m \kappa_i. 
\end{equation*}
where $\alpha_j^ n = \gamma_j \sum_{i=j}^ n \prod_{k=j+1}^ i (1 - \gamma_k)$. Since $|\alpha_j^ n| \lesssim 1$ uniformly for all $j\le n$ (see Lemma 1-(ii) in \cite{polyak1992acceleration}), it is sufficient to show that 
\begin{align*}
    \sup_{1\le m\le n} \left\|\sum_{i=1}^ m  \alpha_i^ m \kappa_i  \right\|  \lesssim \ U_n := \sum_{i=1}^ n \|\kappa_i\| = o_P(\sqrt n) 
\end{align*}
Let $R > 0$ be arbitrary and $T_R = \min \{ \sigma_R, \rho_R \}$. We first bound $ \mathbb E[U_n 1\{T_R \ge n\}]$, and then argue that $ \mathbb E[U_n 1\{T_R <  n\}]$ is relatively small. Recall that on $\{T_R \ge i\}$, $\|W_{i-1}\|\le R$ and $\|(\Phi_{i-1}' W_{i-1} \Phi_{i-1})^{\dagger}\|\le R$. Note also that since $\|\Phi_{i-1} - G\|\le R \eta_{i} \to 0$, the generalized inverse $(\Phi_{i-1}' W_{i-1} \Phi_{i-1})^{\dagger}$ indeed becomes the inverse $(\Phi_{i-1}' W_{i-1} \Phi_{i-1})^{-1}$ for sufficiently large $i$. Thus, on $\{ T_R \ge i \}$, we have that 
\begin{equation*}
    \|\kappa_i\| = \|(\Phi_{i-1}' W_{i-1} \Phi_{i-1})^{\dagger}\Phi_{i-1}' W_{i-1} (\Phi_{i-1} - G) \beta_{i-1}\| \lesssim \|\Phi_{i-1}-G\| \|\beta_{i-1}\|,    
\end{equation*}
for sufficiently large $i$, which implies
\begin{align*}
    \mathbb E[U_n 1\{T_R \ge n\}] &\lesssim \sum_{i=1}^ n \mathbb E [\|\Phi_{i-1} - G\| \|\beta_{i-1}\|1\{T_R \ge i\} ]\\
    &\le \sum_{i=1}^ n \mathbb E [\|\Phi_{i-1} - G\|^ 2]^{1/2} \mathbb E [\|\beta_{i-1}\|^ 21\{T_R \ge i\} ]^{1/2} \\
    &\lesssim  \sum_{i=1}^ n \frac {\gamma_{i}^{1/2}}{\sqrt {i}} \lesssim  \sum_{i=1}^ n i^{-1/2-a/2}
    \lesssim n^{1/2 - a/2} = o(\sqrt n)
\end{align*}
by Cauchy-Schwarz inequality and Part 1 of the proof of Theorem \ref{thm:1}. This shows $U_n 1\{T_R \ge n\} = o_P(\sqrt n)$ for any choice of $R > 0$. 
On the other hand, notice that $\sup_{n\in\mathbb N}\mathbb P(U_n \ne U_n1\{T_R \ge n\}) \le \sup_{n\in\mathbb N}\mathbb P(T_R < n) = \mathbb P(T_R < \infty) \to 0$ as  $R\to \infty$. Thus, for any $\varepsilon>0$, by choosing a sufficiently large $R > 0$, we have for all large $n$,
\begin{align*}
    \mathbb P(U_n /\sqrt n >\varepsilon) &\le \mathbb P(U_n \ne U_n 1\{T_R \ge n\}) + \mathbb P(U_n 1\{T_R \ge n\} /\sqrt n > \varepsilon) \\
    & \le \mathbb P(T_R < \infty) +  \varepsilon^{-1} \mathbb E[n^{-1/2} U_n 1\{T_R \ge n\}]  \\
    & < \frac \varepsilon 2 + \frac \varepsilon 2 = \varepsilon.
\end{align*}
This establishes $U_n = o_P(\sqrt n)$ and completes the proof.
\end{proof}


\bibliographystyle{chicago}

\bibliography{LLSS_bib}


\end{document}